\documentclass[twocolumn, journal]{IEEEtran}

\usepackage[margin=0.5in]{geometry}  
\usepackage{amsmath}
\usepackage{amssymb}
\usepackage{amsthm}

\usepackage[dvips]{graphicx}  
\usepackage{color}  
\usepackage[tight]{subfigure}
\usepackage{bm}
\usepackage{booktabs}
\usepackage[boxruled, linesnumbered]{algorithm2e}
\usepackage{nomencl}
\usepackage[geometry]{ifsym}
\usepackage{color} 

\usepackage{tikz}  
\usetikzlibrary{calc,trees,positioning,arrows,chains,shapes.geometric, %
    decorations.pathreplacing,decorations.pathmorphing,shapes,          %
    matrix,shapes.symbols}

\newcommand{\mc}{\mathcal}
\newcommand{\mb}{\mathbf}

\newcommand{\mbb}{\mathbb}

\usepackage{algorithm2e}

\newtheorem{define}{Definition}   

\newtheorem{lem}{Lemma}

\newtheorem{thm}{Theorem}

\newtheorem{conj}{Conjecture}

\title{Distribution System Outage Detection using Consumer Load and Line Flow Measurements}

\author{\IEEEauthorblockN{
Raffi Sevlian\IEEEauthorrefmark{1}\IEEEauthorrefmark{0},
Yue Zhao\IEEEauthorrefmark{1}\IEEEauthorrefmark{3},
Ram Rajagopal\IEEEauthorrefmark{2}\IEEEauthorrefmark{0},
Andrea Goldsmith\IEEEauthorrefmark{1} \IEEEauthorrefmark{0},  
H. Vincent Poor\IEEEauthorrefmark{3}\IEEEauthorrefmark{0}
}

\IEEEauthorblockA{\IEEEauthorrefmark{1}Department of Electrical Engineering, Stanford University}
\IEEEauthorblockA{\IEEEauthorrefmark{2}Department of Civil and Environmental Engineering, Stanford University}
\IEEEauthorblockA{\IEEEauthorrefmark{3}Department of Electrical Engineering, Princeton University}

\thanks{This research was supported in part by the DTRA under Grant HDTRA1-08-1-0010, in part by the Tomkat Center, in part by Powell Foundation Fellowship, in part by the National Science Foundation award CIF1116377,  in part by the Air Force Office of Scientific Research under MURI Grant FA9550-09-1-0643, and in part by the Office of Naval Research, under Grant N00014-12-1-0767.  } }

\begin{document}
\maketitle

\begin{abstract}
An outage detection framework for power distribution networks is proposed. 
Given the tree structure of the distribution system, a method is developed combining the use of real-time power flow measurements on edges of the tree with load forecasts at the nodes of the tree.
A maximum a posteriori detector {\color{black} (MAP)} is formulated for arbitrary number and location of outages on trees which is shown to have an efficient detector.
A framework relying on the maximum missed detection probability is used for optimal sensor placement and is solved for tree networks.
Finally, a set of case studies is considered using feeder data from the Pacific Northwest National Laboratories.
We show that a 10\% loss in mean detection reliability network wide reduces the required sensor density by 60 \% for a typical feeder if efficient use of measurements is performed. 
\end{abstract}
\section{Introduction}
\label{section:introduction}
Outage detection and management has been a long-standing problem in power distribution networks. 
Outages are caused by protective devices closing off a part of the network to automatically isolate {\color{black} faults}.
Usually, a short circuit fault will trigger this protective operation.
We employ the term \emph{outage detection} to denote the task of finding the status of the protective devices, and the term \emph{fault detection} to denote finding the faults that caused the resulting outage situation.

Many methods for outage and fault detection based on artificial intelligence have been developed.  
Outage detection is often performed prior to fault detection and can greatly improve the accuracy of fault diagnosis.  
For outage detection, fuzzy set approaches have been proposed based on customer calls and human inspection \cite{sumic1996fuzzy}, and based on real-time measurement with a single sensor at the substation \cite{romero2005inference}.  
In networks where supervisory control and data acquisition (SCADA) systems are available, a subset of the protective devices' status can be obtained via direct monitoring. 
When two-way communications from the operator and the smart meters are available, AMI polling has been proposed to enhance outage detection \cite{mak2012synchronizing}. 
There have also been knowledge based systems that combine different kinds of information (customer calls, SCADA, AMI polling) \cite{liu2002knowledge}.  
For fault detection, using only a single digital transient recording device at the substation, fault location and diagnosis systems have been developed based on fault distance computation using impedance information in the distribution system \cite{zhu1997automated}.  
Using only the outage detection results, i.e., the status of the protective relays, expert systems have been applied to locate the underlying faults \cite{fukui1986expert}.  
Incorporating voltage measurements in the distribution system with the outage detection results, fault detection methods based on knowledge based systems have been proposed \cite{balakrishnan1990computer}.  
Fault detection that uses fault voltage-sag measurements and matching has been proposed in \cite{pereira2009improved}, \cite{kezunovic2011smart}.  
Fault diagnosis based on fuzzy systems and neural networks have also been proposed that can resolve multiple fault detection decisions \cite{srinivasan2000automated}.
{\color{black}Existing} outage and fault detection methods based on artificial intelligence do not provide an analytical performance metric, so it is in general hard to examine their optimality.
Their performance can however be evaluated numerically and in simulation studies.
Moreover, because of this lack of an analytical metric, while some of the existing approaches depend on near real-time sensing (e.g. SCADA), they do not provide guidance on where to deploy the limited sensing resources within the distribution system.  

A major alternative to these mechanism is the so called last gasp, where area's in outage will notify, via distress signal that they are out of power.
These provide a duplicate method of outage detection which can be combined with the proposed methods here.
In fact, combining both of these methods can further reduce the time to outage in practical scenarios.

The proposed sensing and feedback framework exploits the combination of \emph{real-time sensing and feedback} from a limited number of power flow sensors and the \emph{infrequent load updates} from AMI or forecasting mechanisms.  
This can is practically possible since there is a growing number of deployments of distribution system line measurements \cite{Sentient2016}, which can measure line current with high precision.

\section{Problem Formulation and Main Contributions}
\label{section-problem-formulation}
Consider a power distribution network that has a tree structure. 
Power is supplied from the feeder at the root, and is drawn by all the downstream loads.  
An outage is a protective device isolating a faulted area.   
When this occurs, the loads downstream of the faulted area will be in outage.
We investigate the optimal design and performance of automatic outage detection systems with the use of the following two types of measurements:  
\begin{itemize}
\item [] \textbf{Noisy Nodal Consumption} typically in the form of forecasts which have forecast errors that must be taken into account.
\item [] \textbf{Error Free Edge Flows} which typically come from real-time SCADA measurements of the power flows on a fraction of the lines.
\end{itemize}

The issue of noisy and error free measurement comes from the fact that loads come from delayed information which needs to be forecasted, while SCADA systems have real time communication potential.
{\color{black} This work assumes lossless power flow, but can also be applied in the case of current measurements on the line and load level.
This can be done, since practical distribution line sensing is accurate in terms of current measurements, and smart meter interval data provides power and voltage information, making current inference possible.}

The main contributions of this work are the following:
\begin{itemize}
	\item [] \textbf{Outage Detection}
	We formulate the problem of detecting any number of possible outages via nodal and edge measurements as a general hypothesis testing problem where the number and locations of outages are unknown.
	We show that this general formulation results in a computationally efficient decentralized hypothesis detector.
	\item [] \textbf{Sensor Placement} 
	We use the decentralized nature of the detector provide a optimal sensor placement with respect to the maximum missed detection error for all hypotheses.
\end{itemize}

\section{System Model and Notation}   
\label{section-System-Model-and-Notation}

\begin{figure}[h]
  \centering
  \includegraphics[scale=0.4]{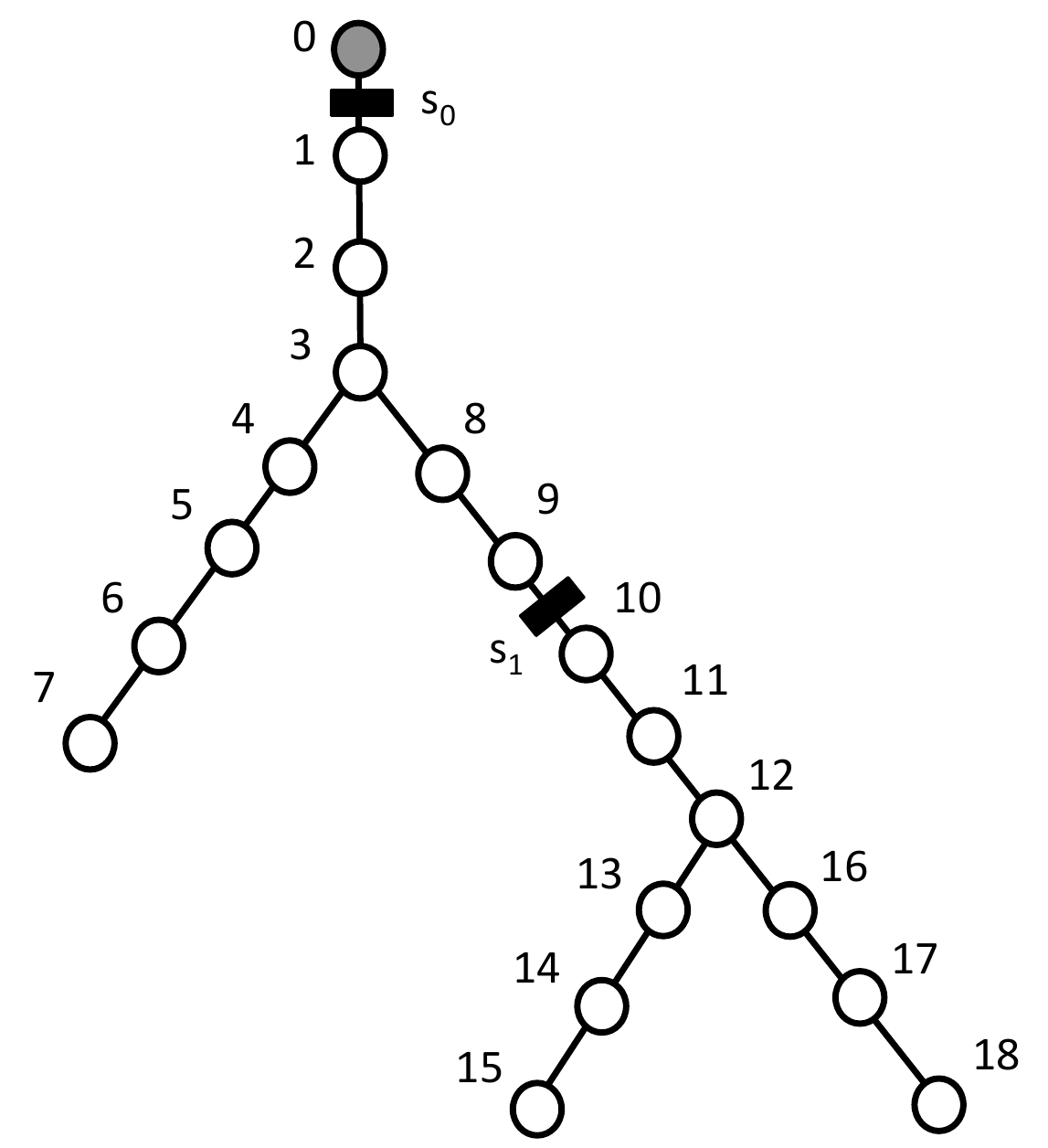}  
  \label{fig:example_tree_10_node}
  \caption{
  Example tree $\mc{T}_1$ used to illustrate various properties.
  Each node in the network is numbered.  
  Node $v_n$ is connected to it's parent via edge $e_n$ consuming $x(v)$ power at each node.
  Two flow measurement sensors $s_0, s_1$ along with load pseudo measurements $\hat{x}(v)$. 
  }
  \label{fig:example_tree_10_node}
\end{figure}

\smallskip
\noindent\textit{Topology of the Distribution System}:
The vertices in the distribution network are indexed by $V = \{v_0,v_1,\ldots,v_N\}$, with bus $v_0$ denoting the root of the tree. 
We index by $e_n$ the line that connects bus $v_n$ and its parent node.  

\smallskip
\noindent\textit{Outage Hypothesis Model}: 
Outages are modeled as disconnected edges corresponding to protective devices disconnecting loads on a network. 
For example, consider single line outages in a tree with $N$ edges: 
In this situation, there will exist $N$ single edge outage hypotheses and a single non-outage situation.
Let $\mc{H}^{1} = \{e_1, \hdots, e_{N} \cup \emptyset \}$ be the set of all single outage hypotheses for a tree $\mc{T}$.

We consider a more general case of an unknown number and location of potential outages.
{\color{black} We define the set of up to k edge outages $\mc{H}^{k}$ as the set of k edge hypothesis.}
This set follows:
\begin{align}
\mc{H}^{k} = \underbrace{ \mc{H}^{1} \times \mc{H}^{1} \hdots \mc{H}^{1} }_{\text{k times}}.
\end{align}

   
\smallskip
\noindent\textit{Load Model:}
Each node $v$ in the graph has a consumption load $x(v)$.
The forecast of each load is $\hat{x}(v)$ with error $\epsilon(v) = x(v) -\hat{x}(v)$.   
We assume errors are mutually independent random variables that follow $\epsilon(v) \sim N(0, \sigma(v)^2)$.  
Given the forecasts we treat the true load, which is unknown to us, as a random variable $x(v) \sim N( \hat{x}(v), \sigma^2(v))$.
In the vector case, we have
\begin{align}
\hat{\mb{x}} \sim N(\mb{x}, \Sigma) \label{eq:forecast-model}
\end{align}
where we can assume $\Sigma$ is a diagonal covariance matrix.

\smallskip
\noindent\textit{Measurement Model:} 
For any edge $e$,  denote by $s$ the power flow on it towards all active downstream loads.
The measured flow depends on the network topology, outage situation and the true loads. 
The sensor placement is denoted as $\mc{M}$ with $\mc{M} \subset E$. 
The vector of all measurements is $\mathbf{s} \in R^{|\mathcal{M}|}$.

Given a tree $\mc{T}$ assume hypothesis $H$ corresponds to the outage of any number of disconnected edges.
The measured power consumption of the $i^{th}$ sensor measurement under hypothesis $H \in \mc{H}^{k}$ is 
\begin{align}   
s_i(H) = \sum_{v \in {\color{black} V_{i}(H) } } x(v), \label{eq-flow-measurement} 
\end{align}  
where the set {\color{black} $V_i(H)$} indicate the set of vertices to be summed over under any particular hypothesis.

A general representation of the observed flow is the following.
The set of full flow observations $\mb{s}$, given a particular hypothesis $H \in H^{k}$, we can represent the observations as:
\begin{align}
\mb{s} = \Gamma_{H} \mb{x}~\forall H \in \mc{H}^{k} \label{eq-matrix-flow-relation},
\end{align}
where $\Gamma_{H} \in \{0, 1 \}^{ |\mc{M}| \times |V| }$.
Here $\Gamma_{H}$ is generated for each hypothesis and we assume the forecast error covariance $\Sigma$ can be estimated from the load forecasting process.
      
\section{General Outage Detection}
\label{subsection-Detection-Detection}

Consider the general outage detector.
Given the vector of load forecasts, $\hat{\mb{x}}$, nominal forecast error $\Sigma$ and real time load flows $\mb{s}$ along a set of branches, the detector must determine the correct number and location of each edge in outage $H \in \mc{H}^{k^{\star}}$. 

{\color{black} These are single snapshot values of load forecast and line flow.
A multi period detection framework can be analyzed in a similar fashion.}
We first construct a simple but naive multiple hypothesis detector relying on a maximum likelihood estimator.
Consider the flow model in eq. \eqref{eq-matrix-flow-relation}, relating the true load at each node to the observed flow on the network.

Given the forecast model in eq. \eqref{eq:forecast-model} and the hypothesis model in eq. \eqref{eq-matrix-flow-relation}, the Maximum a Posteriori detector is the following:
\begin{align}
\{k^{\star},~\hat{H} \} &  \in \underset{ H \in \mc{H}^{k} }{\arg\max}  \Pr\left( \mb{s}~|~\hat{\mb{x}},~H \right) \label{general-ML}
\end{align}
See Appendix \ref{subsection-MAP-Detection-for-Outage-Hyptheses}, for details.

The flow likelihood can be computed as follows:
\begin{align}
\mb{s} | \{ \hat{\mb{x}},~H \} &= \Gamma_{H} \mb{x} \nonumber \\
                &= \Gamma_{H} ( \mb{\hat{x}} + \mb{\epsilon} ) \nonumber \\
                &\sim N \left( \Gamma_{H} \mb{\hat{x}}, \Gamma_{H} \Sigma \Gamma^{T}_{H} \right)~~\forall H \in \mc{H}^{k} \label{eq-flow-forecast-prob} 
\end{align}
Eq. \eqref{eq-flow-forecast-prob} allows us to evaluate a likelihood under each possible hypothesis.
{\color{black} Therefore, a naive detector will enumerate every possible outage, evaluate it's likelihood, and choose the maximum.}
This is difficult for the following reasons:
\begin{enumerate}
\item The set $\mc{H}^{k}$ is of size ${ |E| \choose k }$, so computing the set $\mc{H}^{k}$ can be very expensive.  
Enumerating the entire maximum likelihood detector requires $\sum_{k=1} { |E| \choose k }$ evaluations.
\item Many of the potential hypotheses map to  the same observed flows, therefore the detector output is not unique.
{\color{black} This occurs when one edge is a descendant of another.} 
\item {\color{black} Missed detection errors in a multivariate hypothesis testing framework can only be evaluated via monte carlo testing.
There is no insight in optimizing placement.}
\end{enumerate}  
   
\section{Decoupled Maximum Likelihood Detection} 
\label{section-Decoupled-Maximum-Likelihood-Detection} 

{\color{black} We show that the issues regarding the general maximum likelihood detector can be overcome by decoupling the hypotheses and the observations, given the tree structure of the outage detection problem.} 
This leads to a simple decentralized detector which is equivalent to eq. \eqref{general-ML}, where each decision is a scalar hypothesis test.
This leads to an efficient hypothesis enumeration, detection and error evaluation.

In the following sections, we show the following:
\begin{enumerate}
\item The original search space $\mc{H}^{k}$ can be replaced by a set $\mc{H}_{u}$ of uniquely detected outages, due to the tree structure of the network.
\item {\color{black} Processing zero/positive flow information reduces the search space from $\mc{H}_{u}$ to $\mc{H}^{+}_{u}$.
which decouples into a product set of local hypotheses: $\mc{H}^{+}_u = \underset{ A \in \mc{A}^{+} }{\bigwedge} \mc{H}_u^{+}(A)$.
Where, $A$ indicates a local area with it's own hypothesis set.}
\item {\color{black} The joint likelihood function $\Pr ( \mb{s}~|~\hat{\mb{x}},~H )$ decouples along each area.}
\end{enumerate}
Combining these results leads to a decentralized detector which can be solved easily.

\subsection{Unique Outages} 
\label{subsection-unique-outages}

{\color{black} Maximizing the likelihood of observations over the k outage set can lead to a non-unique solution.}
An alternative is to only consider uniquely detectable outages, where no possible outage event is downstream of any other.

Define the set of outage hypotheses $\mc{H}_{u}$ (u is for unique) as follows:
{\color{black}
\begin{align}
\mc{H}_{u} =& \{ H \in \mc{H}^{k}  \text{ for some k, s.t no two edges} \nonumber \\ 
                    & \text{ are descendant of each other.} \}
\end{align}
This definition is not constructive, but useful.}
Consider tree $\mc{T}_2$ shown in Figure \ref{fig:Hu_simple_example_graph}.
Here $\mc{H}_{u}$ can be enumerated by simple observation.
\begin{align}
\mc{H}_u &= \{ \emptyset,~e_1,~e_2,~e_3,~e_4,~e_5,~(e_3 \times e_5),~(e_4 \times e_5)  \} \label{T2-example}
\end{align}
There is a single non-outage hypothesis $\emptyset$, and 5 single outage hypotheses, and 2 double outage hypotheses.

For a more general case (Figure \ref{fig:example_tree_10_node}), enumerating this set for a tree can be performed recursively. 
The set $\mc{H}_u$ can be enumerated using a 'branch-network' based on the original tree.
The tree in Figure \ref{fig:example_tree_10_node} is depicted in Figure \ref{fig:Hu_recursive_graph_network} with nodes and edges removed which highlights the various branches of the graph.
Each set of branches are aggregated as a node to be traversed, in the hypothesis enumeration procedure.

Consider a {\color{black} set} function $E(b) = \{e \in E : e \text{ is along branch } b_i \}$ to enumerate the set of edges on a branch.
Given the two examples, we have the following branch-edges:
\begin{itemize}
\item $\mc{T}_1$:
$E(b_{1}) = \{ e_{1}, e_{2}, e_{3} \}$, $E(b_{2}) = \{ e_{4}, e_{5}, e_{6}, e_{7}  \}$, $E(b_{3}) = \{ e_{8}, e_{9}, e_{10}, e_{11}, e_{12} \}$, $E(b_{4}) = \{ e_{13}, e_{14}, e_{15} \}$
and $E(b_{5}) = \{ e_{16}, e_{17}, e_{18} \}$.
\item $\mc{T}_2$: $E(b_{1}) = \{ e_{1}, e_{2} \}$, $E(b_{2}) = \{ e_{3}, e_{4} \}$, $E(b_{3}) = \{ e_{5} \}$.
\end{itemize}
\begin{figure}[h]
\hspace{-5mm}
\subfigure[][]{
\includegraphics[scale=0.3]{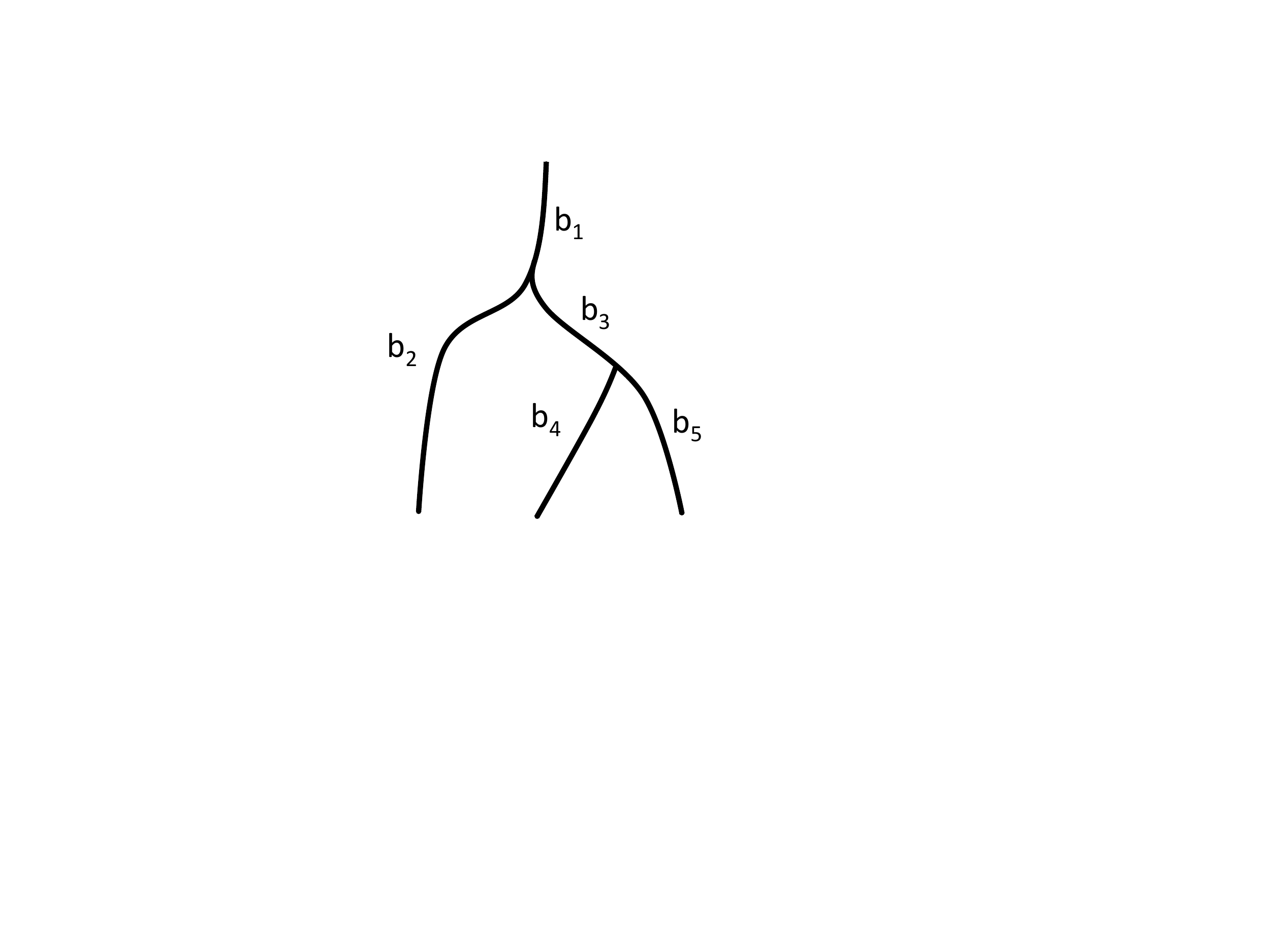}  
\label{fig:Hu_recursive_graph_network}
}
\subfigure[][]{
\includegraphics[scale=0.2]{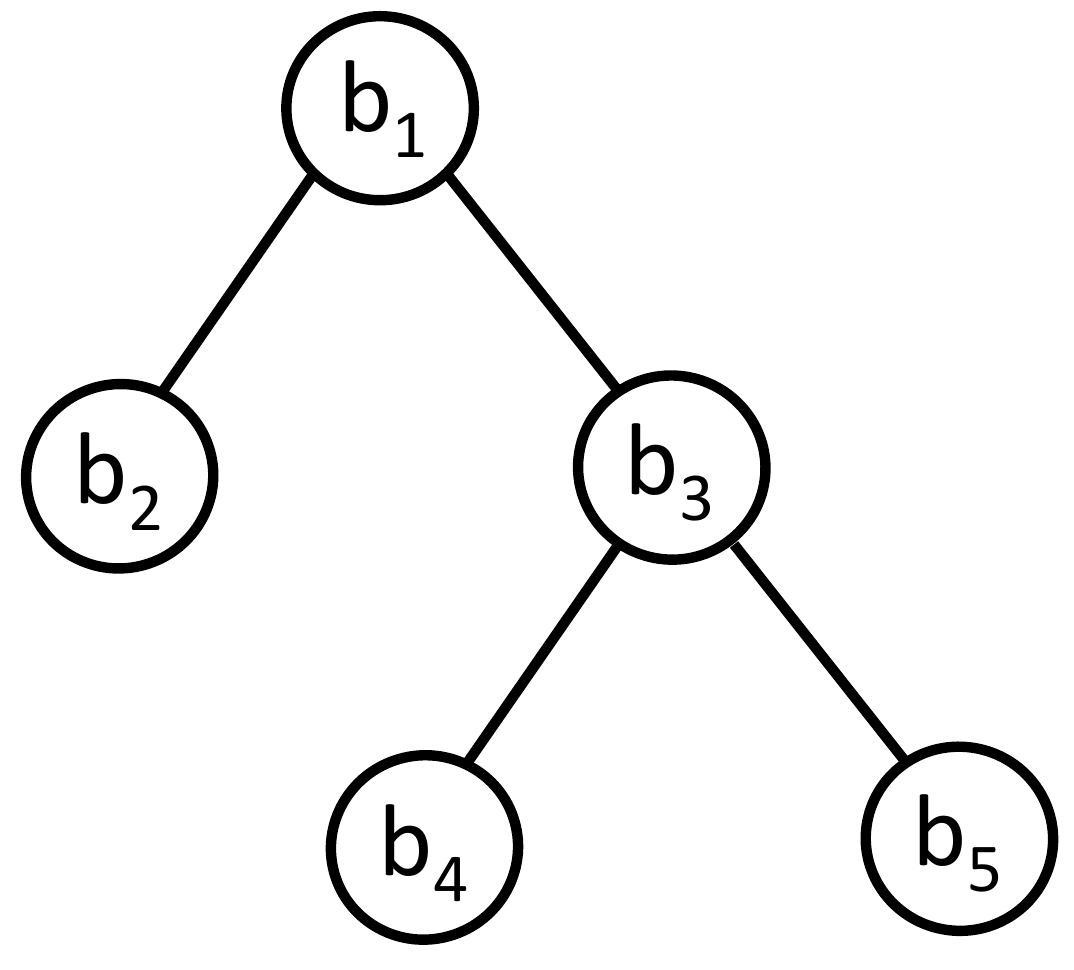}  
\label{fig:Hu_recursive_branch_network}
}
\subfigure[][]{
\includegraphics[scale=0.5]{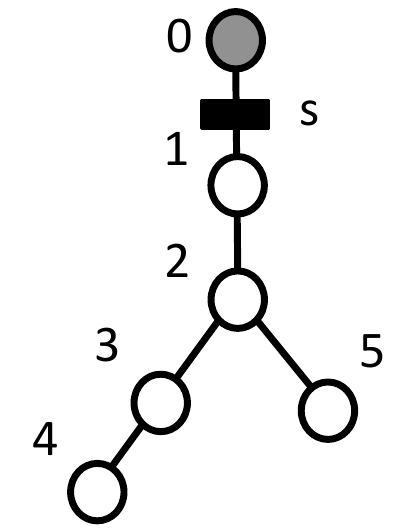}  
\label{fig:Hu_simple_example_graph}
}
\subfigure[][]{
\includegraphics[scale=0.25]{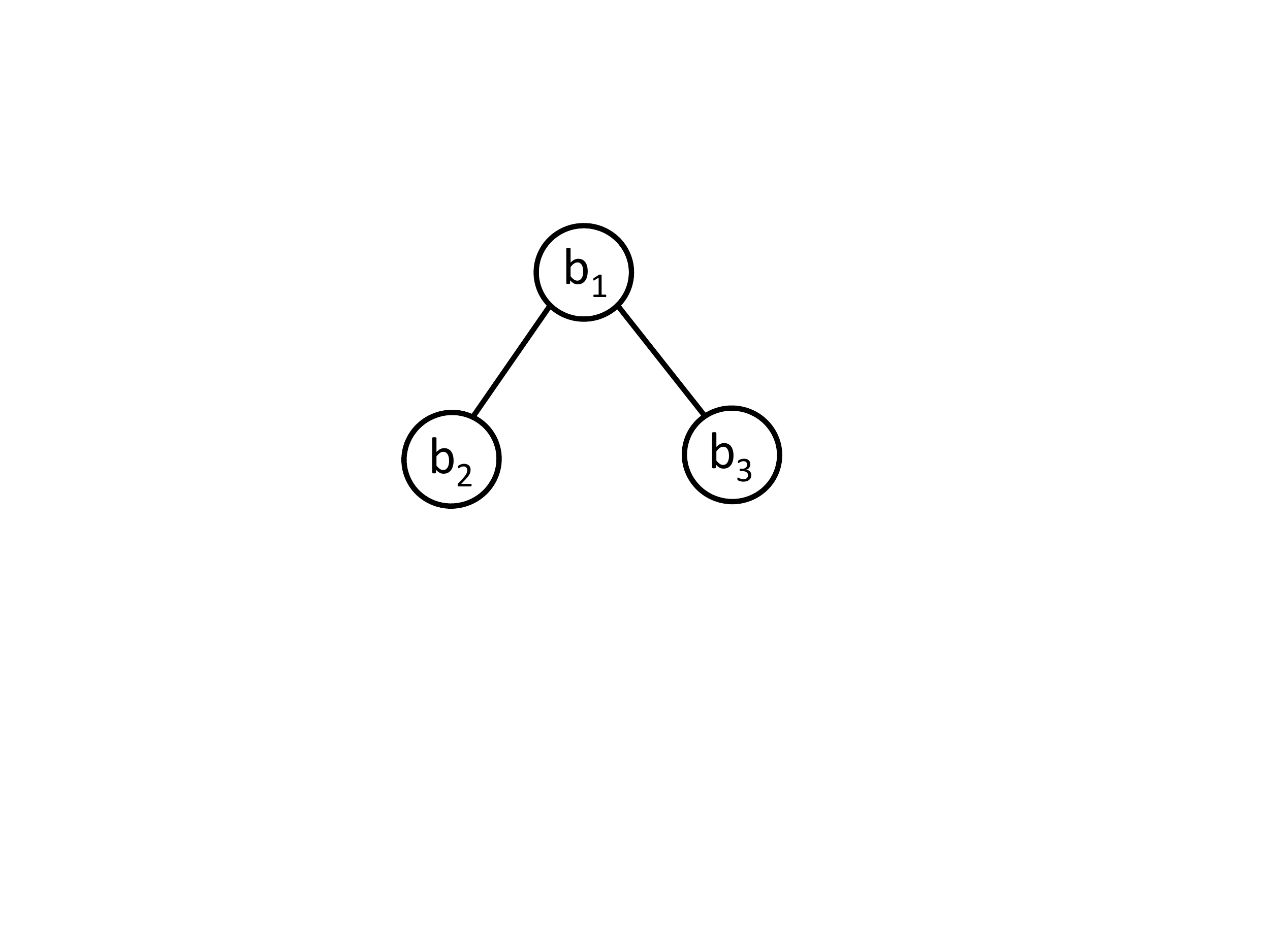}  
\label{fig:Hu_simple_example_branch_graph}
}
\caption{   
\ref{fig:Hu_recursive_graph_network}   Branches of tree $\mc{T}_1$.
\ref{fig:Hu_recursive_branch_network} Branch network for tree $\mc{T}_1$.
\ref{fig:Hu_simple_example_graph} Simple tree network $\mc{T}_2$.
\ref{fig:Hu_simple_example_branch_graph} Branch graph for $\mc{T}_2$. }
\end{figure}
Using this definition, we propose the following recursive definition of the set $\mc{H}_u(b)$, which indicates the set of hypotheses formed from branch $b$ and all descendants.
From this definition, it is clear that $\mc{H}_u \equiv \mc{H}_u(b_1)$, since this is the root branch of the tree.
\begin{align}
\mc{H}_u(b) = E(b) \cup \left(  \bigcup_{ \mb{b} \in \mbb{P}( \textbf{child}(b) ) }   \left(  \underset{b \in \mb{b}} { \bigwedge } \mc{H}_{u}(b)  \right)  \right) \label{eq-recursive-Hu}
\end{align}
Here $b$ is the current node, and $\textbf{child}(b)$ is the set of children of $b$ in the branch tree.
The set $\mbb{P}( \textbf{child}(b) )$ is the power set of all the child branches, where any element of the power set is $\mb{b}$.
Note that if a branch has no descendants, we merely evaluate $E(b)$, since the remaining terms are null.
Eq. \eqref{eq-recursive-Hu} is quite unwieldy, but can be interpreted easily.
For tree $\mc{T}_2$ in Figure \ref{fig:Hu_simple_example_graph} the child branches are $\textbf{child}(b_1) = \{b_2, b_3 \}$ while the power set is:
\begin{align}
\mbb{P}( \{b_2,  b_3\} ) = \{ \underbrace{\{ \emptyset \}}_{\mb{b}_0}, \underbrace{\{  b_2  \}}_{\mb{b}_1}, \underbrace{ \{  b_3  \}  }_{\mb{b}_2},  \underbrace{ \{b_2, b_3 \}  }_{\mb{b}_3} \}.
\end{align}
{\color{black}
Evaluating \eqref{eq-recursive-Hu}, we arrive at the following:
\begin{align}
 \mc{H}_u(b_1) &= E(b_1) \cup \left( \underset{b \in  \{ \emptyset \} } { \bigwedge } \mc{H}_{u}(b) \right) \cup \left( \underset{b \in \{ b_2 \} } { \bigwedge } \mc{H}_{u}(b) \right)   \nonumber  \\ 
                		& ~~~~~\cup \left( \underset{b \in \{ b_3 \} } { \bigwedge } \mc{H}_{u}(b) \right) \cup \left( \underset{b \in  \{ b_2, b_3 \} } { \bigwedge } \mc{H}_{u}(b) \right)   \label{eq:recursion-T2-A} 
\end{align}  
We rely on the following definitions:
\begin{itemize}
\item [D1] Base case: $\underset{b \in  \{ \emptyset \} } { \bigwedge } \mc{H}_{u}(b) = \emptyset$.
\item [D2] Hypotheses double counting: $e_i \cup e_i = e_i$, including the empty set $\emptyset \cup \emptyset = \emptyset$.
\item [D3] Cross product reduction: $\emptyset \times e_i = e_i$, which implies $E(b) \times \emptyset = E(b)$.
\end{itemize}
For example, eq. \eqref{eq-recursive-Hu} is applied to $\mc{T}_2$ as follows:
For $b_1$, we have that 
\begin{align}
\mc{H}_u(b_1) = \{ E(b_1) \cup \mc{H}_{u}(b_{2}) \cup \mc{H}_{u}(b_{3})  \cup \mc{H}_{u}(b_{2}) \times \mc{H}_{u}(b_{3})   \}. 
\end{align}
For the branches $b_{2}$ and $b_{3}$ we use eq. \eqref{eq-recursive-Hu} and D1 resulting in $\mc{H}_{u}(b_{2}) = \{E(b_2) \cup \emptyset \}$ and $\mc{H}_{u}(b_{3}) = \{E(b_3) \cup \emptyset \}$.
Using D2 and D3, we have the following:
This results in:
\begin{align}
 \mc{H}_u(b_1) &=  \{ E(b_1) \cup \{ \emptyset \} \cup  \nonumber \{ E(b_2) \cup \emptyset \} \cup \{ E(b_3) \cup \emptyset \} 			\nonumber\\
                         & ~~~~ \cup \{ E(b_2) \cup \emptyset \} \times \{ E(b_3) \cup \emptyset \}  \}   \label{eq:HuB3}      					\nonumber \\ 
			&=  \{ \emptyset \cup E(b_1)  \cup \{ E(b_2) \cup E(b_2) \cup E(b_3)  \cup \left( E(b_2) \times E(b_3) \right) \}    \nonumber    \\ 
                		&= \{ \emptyset,~e_1,~e_2,~e_3,~e_4,~e_5,~(e_3 \times e_5),~(e_4 \times e_5)  \} 						\nonumber
\end{align}
}  
Which is identical to simple enumeration.     
At this point, we have reduced the maximum likelihood detector to the following form:
\begin{align}
\hat{H} = \underset{ H \in \mc{H}_{u} }{\arg\max} \Pr\left( \mb{s}~|~\hat{\mb{x}},~H \right)
\end{align}  
Note the equality in the maximization since the restriction of the search space allows us to find a unique solution.  
Notice that this procedure will automatically enumerate all possible number of outages as well as their positions in the graph. 
Under the general model of any number of edge outages, this is the complete enumeration of the hypothesis set which has a high computational burden.

In the following sections we show that (1) the search space decouples to a product space of local hypotheses given binary flow indicators and (2) the likelihood function decouples along these local 'areas', thereby reducing to a decoupled scalar hypothesis test for each local area.

\subsection{Reducing Hypotheses from binary flow information}
\label{subsubsection-Reducing-Hypotheses-from-Zero-Flows}  

\begin{figure}[h]
\centering   
\subfigure[][]{
\includegraphics[scale=0.4]{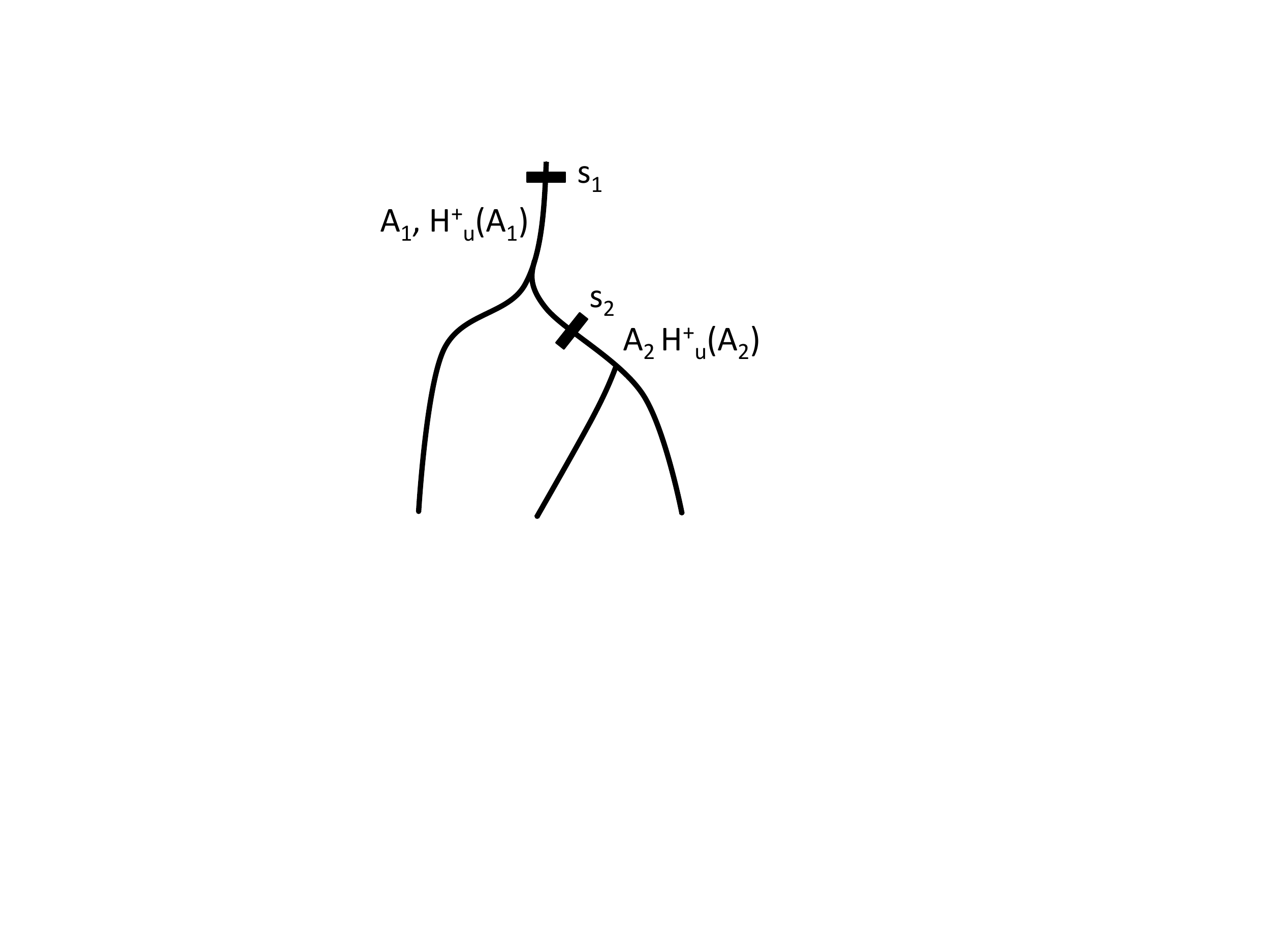} 
\label{fig:multiple_area_network}    
}
\subfigure[][]{
\includegraphics[scale=0.21]{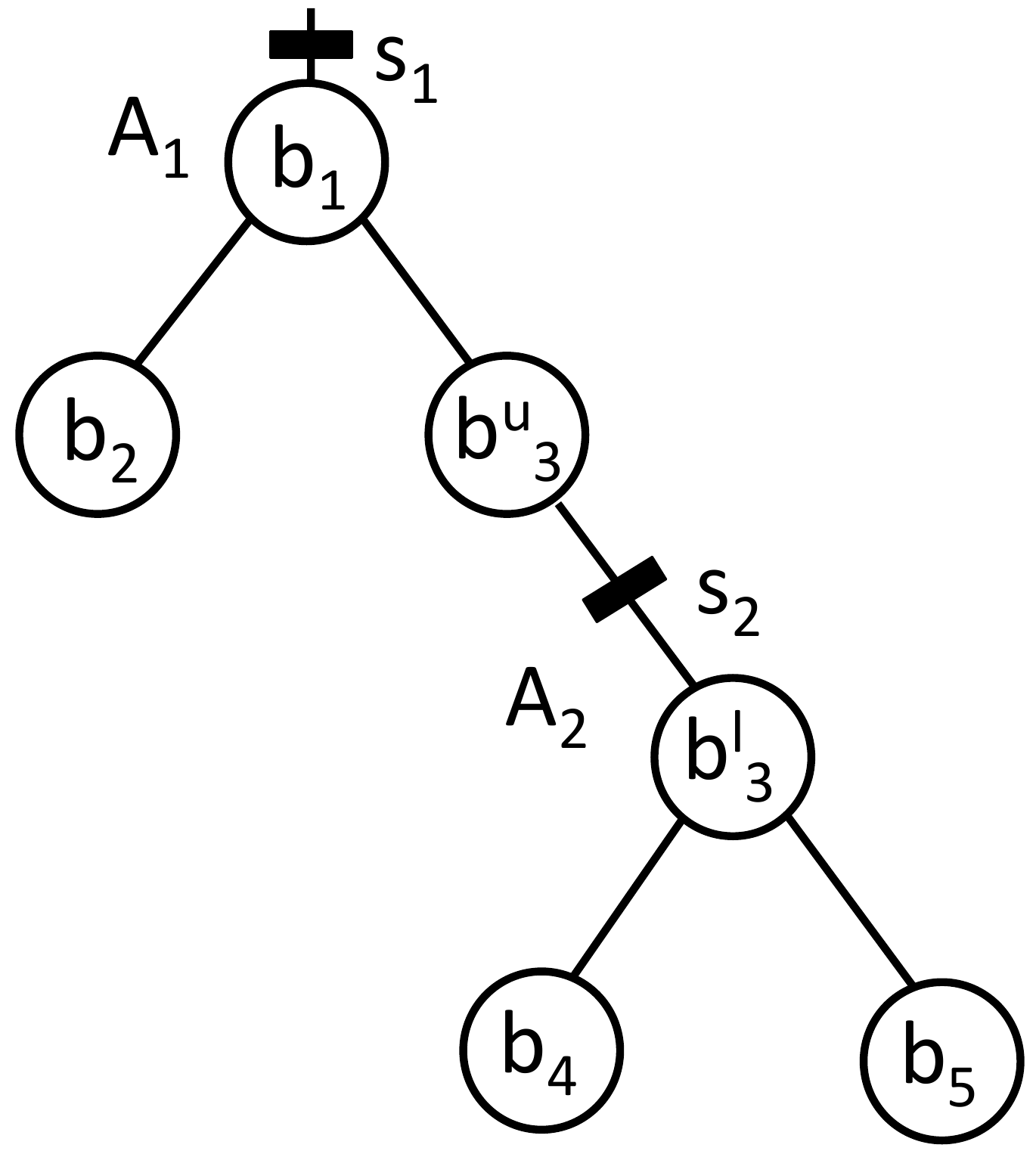}  
\label{fig:Hu_recursive_branch_network_with_flow}
}
\subfigure[][]{  
\includegraphics[scale=0.33]{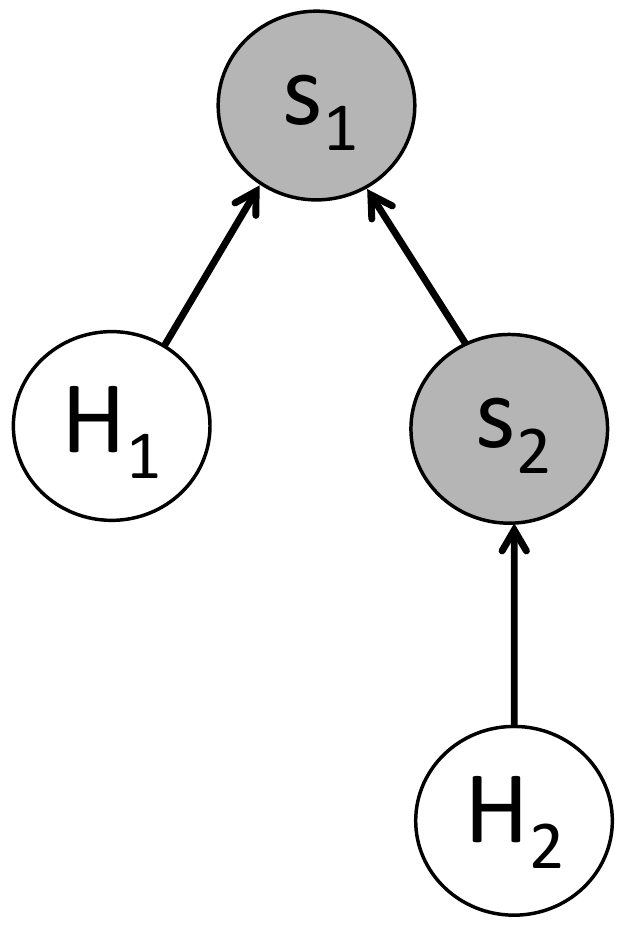} 
\label{fig:outage_dbn_model}    
}  
\caption{
\ref{fig:multiple_area_network} The branch network for network $\mc{T}_1$ with associated sensors $s_1$, $s_2$ and area hypotheses {\color{black} $\mc{H}^{+}_u(A_1)$ and $\mc{H}^{+}_u(A_2)$.}
\ref{fig:Hu_recursive_branch_network_with_flow} 
{\color{black} Each area will keep the local branches.
Branch $b_3$ is split between two areas and processed as different vertices in branch-network ($b^{u}_3$ and $b^{u}_3$}.
\ref{fig:outage_dbn_model} The network can be modeled by a directed graphical model indicating observations (shaded) and variables that must be maximized over.
}
\end{figure}  
The detection problem will encounter a set of positive flows as well as zero's when a sensor is downstream of an edge outage.
The ML detection observations and search space can be reduced as:
\begin{align}
~\hat{H} &= \underset{\forall H \in \mc{H}_{u} }{\arg\max} \Pr\left( \{ \mb{s}_{z}~\mb{s}_{+} \} ~|~\mb{\hat{x}},~H \right) \\
              &= \underset{\forall H \in \mc{H}^{+}_{u} }{\arg\max} \Pr\left( \{ \mb{s}_{+} \} ~|~\mb{\hat{x}},~H \right)
\end{align}
By processing the flow information we reduce the set of hypothesis in the detector from the set $\mc{H}_{u}$ to the reduced set $\mc{H}^{+}_{u}$.
More importantly, we show that $\mc{H}^{+}_{u}$ does not require a recursive enumeration, but actually decouples as a product space of local hypotheses.
This is first shown in an example, then the general form is stated.
Consider $\mc{T}_1$ with branch network with flow measurements in Figure \ref{fig:Hu_recursive_branch_network_with_flow}.
In this example, all branches are unchanged except $b_3$, which is split into $b^u_3$ and $b^l_3$ and is separated by flow measurement $s_1$.
In the case of splitting branch nodes, the upper branch will take the edge with the measurements so $E(b^u_3) = \{e_8, e_9, e_{10} \}$ and $E(b^l_3) = \{e_{11}, e_{12} \}$.

{\color{black} The following illustrative example highlights a general decoupling principle of the hypothesis set $\mc{H}^{+}_u$.}
Let's consider the two cases separately {\color{black} ($s_1, s_2 > 0$ and $s_1 > 0, s_2 = 0$)} which provides the intuition for the general case.
\subsubsection{Case 1}{\color{black} $(s_1 > 0$, $s_2 > 0)$}
This implies that \textit{there cannot be any outage with edges in $b_1$ or $b_3$}, else $s_2 = 0$.
A brute force enumeration of $\mc{H}_u^{+}(b_1)$ is done by enumerating $\mc{H}(b_1)$ according to Figure \ref{fig:Hu_recursive_branch_network}, then removing any terms with $E(b_1)$ and $E(b_3)$ outages.
It can be shown that:
\begin{align*}
\mc{H}_u(b_1) & = \{ \emptyset \cup E(b_1) \cup  E(b_2) \cup ( E(b^{u}_3) \cup \mc{H}_u(b^{l}_3))  \\
		       & ~~~~~~~~~\cup E(b_2) \times ( E(b^{u}_3) \cup \mc{H}_u(b^{l}_3) ) \} 
\end{align*}
Removing the possible outages due to the positive flow information (i.e. any tuple with edges in $E(b^{u}_3)$, or $E(b_1)$) , we have:
\begin{align*}
\mc{H}_u^{+}(b_1) &= \{ \emptyset \cup  E(b_2) \cup \mc{H}_u^{+}(b^{l}_3) \cup \left( E(b_2) \times \mc{H}_u^{+}(b^{l}_3) \right) \}  \\  
                              &=\{ \emptyset_1 \cup  E(b_2) \} \times \{ \emptyset_2 \cup  \mc{H}_u^{+}(b^{l}_3) \}                                         \\ 
                              &= \mc{H}_u^{+}(A_1) \times \mc{H}_u^{+}(A_2)  							                                           
\end{align*}

{\color{black}
We use $A_1$ and $A_2$ to define a local area.
An area is a partition of the original tree, which will decouple the set of all hypotheses into a product space of 'area hypotheses'.
Each area will contain a root sensor (ex. $A_1$ contains $s_1$) and a set of descendent sensors (ex. $\textbf{child}(s_1) = \{ s_2 \}$).

Within a local area, a set of unique outage hypotheses are evaluated $\mc{H}_u^{+}(A_1)$ which satisfy the binary observations of whether any flow is observed along the downstream sensors $\textbf{child}(s_1) >0$.  
The local hypotheses of each area are later combined to form any possible hypothesis from the original enumeration of unique outages.
Appendix \ref{subsection-Extended-Discussion-of-Recursive-Evaluation-of-H_u} contains a more detailed discussion.
}

\subsubsection{Case 2} {\color{black} $(s_1 > 0$, $s_2 = 0)$}
This implies that \textit{all possible outages must contain $b_1$ or $b^u_3$}, else $s_2 > 0$.
First we enumerate $\mc{H}_u(b_1)$ then keep only those elements which lead to $s_2 = 0$ (i.e. every tuple with edges in $E(b^{u}_3)$, or $E(b_1)$) , we have:
\begin{align*}
\mc{H}_u^{+}(b_1) &= \{\emptyset \cup E(b_1) \cup E(b^{u}_3) \cup E(b_2) \times E(b^{u}_3) \} \\
 			 &= \{ \mc{H}_u^{+}(A_1) \} 
\end{align*}
Note that in this case, $\mc{H}_u^{+}(A_2)$ is never evaluated since $s_2 = 0$.
Different positive and zero flow patterns lead to changes in the local hypothesis set $\mc{H}_u^{+}(A)$.

\subsubsection{General Hypothesis Decoupling}
\label{general-hypothesis-decoupling}
{\color{black} A complete description of the set should be $\mc{H}_u^{+}(A, f)$, where {\color{black} $f = \mbb{I}_{\{\textbf{child}(s_i) > 0\} }$}, since this set depends on the binary flow information of the child sensors.
For example Case 2 is $\mc{H}^{+}_u = \mc{H}_u^{+}(A_1, \{1~0\})$ and Case 2 is $\mc{H}^{+}_u = \mc{H}_u^{+}(A_1, \{1~1\}) \times \mc{H}_u^{+}(A_2, \{1~1\})$. }
For any arbitrary tree and flow sensors, the unique hypothesis set conditional on binary flows $\mc{H}_u^{+}$ will decouple according to:
\begin{align}
\mc{H}_u^{+} = \underset{ A \in \mc{A}^{+} }{\bigwedge} \mc{H}_u^{+}(A). \label{eq:general-hypothesis-decoupling}
\end{align}
Where the set $\mc{A}^{+}$ indicates the areas which have a root measurements $s_i > 0$.
Finally, $\mc{H}_u^{+}(A)$ is the local conditional hypothesis, which can be computed in the general case.
{\color{black}
Any given $H_{i} \in \mc{H}^{+}_u$ can be represented by a product of area hypotheses $H_{i} = H_{1, i(1)} \times \hdots \times H_{M, i(M)}$, where for the $k^{th}$ area $H_{k, i(k)} \in \mc{H}_u^{+}(A_k)$.
Index $i(k)$ is the particular index into the $k^{th}$ hypothesis set corresponding to global hypothesis $H_i$.
Appendix \ref{subsubsection-general-hypothesis-decoupling} presents a general algorithm generate the set $\mc{H}_u^{+}(A, f)$, for arbitrary binary flow information.
}
\subsection{Decoupling the Joint Likelihood:}
\label{subsection-Decoupling-the-Joint-Likelihood:}

Due to the model of noiseless flow and forecasted nodal measurements joint likelihood of all observations, given $\hat{\mb{x}}$, and $H \in \mc{H}^{+}_u$, decouple across the areas.
In general we have: 
\begin{align}		
\Pr\left(\mb{s}_{+}~|~\mb{\hat{x}},~H \right) = \prod_{i: A_i \in \mc{A}^{+} } \Pr\left( s_i~|~\textbf{child}(s_i),~\mb{\hat{x}},~H_i  \right) 
\end{align}	
The tree $\mc{T}_1$ is presented here as an example, while the general case can be easily shown as an extension of the example.

Consider again the sub-graph in Figure \ref{fig:multiple_area_network}, where we have already reduced the hypotheses using binary flow information (both $s_0, s_1 > 0$).
{\color{black} 
As discussed in Section \ref{general-hypothesis-decoupling}, a unique hypothesis $H_i \in \mc{H}_u^{+}$ can be represented as $H_i = H_{1, i(1)} \times H_{2, i(2)}$.
Where $H_{1, i(1)} \in \mc{H}(A_1)$ and $H_{2, i(2)} \in \mc{H}(A_2)$.

From eq. \eqref{eq-flow-measurement} the observed flow $s_1$, can be computed as follows:
\begin{align}
s_1 | \{\hat{\mb{x}},~H_{i} \} 
	  &= \sum_{v \in V(H_{1, i(1)} \times H_{i, i(2)}) } x(v)                    				                                                     \label{tree2-s1-decoupling-line1}  \\
	  &= \sum_{v \in V(H_{1, i(1)} \times H_{2, \emptyset}) \setminus V(H_{2, \emptyset}) } x(v) + \sum_{v \in V(H_{2, i(2)}) } x(v) \label{tree2-s1-decoupling-line2}  \\
 	  &= \sum_{v \in V(H_{1, i(1)}) \setminus V(H_{2, \emptyset} )} x(v) + s_2                                                                        \label{tree2-s1-decoupling-line3}  
\end{align}	
Where:
\begin{itemize}
\item Eq. \eqref{tree2-s1-decoupling-line1} is the summation of true loads.
By decoupling of the hypotheses across areas, we represent $H_i$ as the product of the two local hypotheses $H_{1, i(1)}$ and $H_{2, i(2)}$.
\item In \eqref{tree2-s1-decoupling-line2}, this can be separated as the sum of parts
\begin{enumerate}
\item $V(H_{1, i(1)} \times H_{2, \emptyset}) \setminus V(H_{2, \emptyset})$ are the vertices in the summation in $A_1$ independent of what is happening of downstream areas.
$H_{2, \emptyset}$ indicates the non-outage hypothesis in area 2.
\item $V(H_{2, i(2)})$ is the set of vertices in the summation for area Area 2.  
Although $H_{2, i(2)}$ is unknown, $s_2 = \sum_{v \in V(H_{2, i(2)}) } x(v)$, the the unknown hypothesis can be eliminated.
\end{enumerate}
\item Eq. \eqref{tree2-s1-decoupling-line3}, replaces the second summation with the observed flow, since they are equivalent.
Therefore, if we condition on the remaining flow observation, $\textbf{child}(s_1)$, the flow decouple between different areas.
\end{itemize}
  
Finally, since $x(v)$, is not known, a \textit{likelihood function for the net flow in the area} can be constructed given the load forecasts.
The first term in eq. \eqref{tree2-s1-decoupling-line3} is modeled with the following:
\begin{align}
s_1 - s_2 | \{s_2~\hat{\mb{x}},~H \} \sim N(\mu(\mb{\hat{x}}, H ), \sigma(\mb{\hat{x}}, H)).
\end{align}

Evaluating $\mu(\mb{\hat{x}}, H)$, and $\sigma(\Sigma, H)$ can be computed easily and is described in Section \ref{subsection-Evaluating-P_MAX-for-tree-networks}.

Similarly, $s_2$ is decoupled from any hypothesis in $\mc{H}_u^{+}(A_1)$, since the measurement depends only on downstream hypotheses (assuming $s_2$ has some positive flow to begin with).
}
In this example we can decouple the likelihood function as follows.
\begin{align}		
\Pr\left(s_1~s_2 |~\mb{\hat{x}},~H_1,~H_2 \right)  = \Pr\left(s_2 | ~\mb{\hat{x}},~H_2 \right)\Pr\left(s_1 | s_2~\mb{\hat{x}},~H_1 \right) \nonumber
\end{align}		
In this example, this decoupling can be represented as a simple graphical model as shown in Figure \ref{fig:outage_dbn_model}, where each local hypothesis is an unknown variable that must be determined via likelihood maximization.
Conditioning on the the only observations $s_1$, $s_2$, the two hypotheses variables are independent.
{\color{black}
This graphical model formulation is used to show the general case where there may be noise in the flow measurements.
In such a case, the decoupling would not work, but a message passing algorithm can be used.
The graphical model formulation can also be applied in different sensor types.
To show the general case, eq \eqref{tree2-s1-decoupling-line1}-\eqref{tree2-s1-decoupling-line3}, can be extended to a general area network with multiple downstream sensors.
}
\subsection{Decoupled Maximum Likelihood Function:}   
\label{subsection-Decoupled-Outage-Detector:}

We can combine the  results shown so far to a decoupled likelihood function.
\begin{align}		
\{k^{\star},~\hat{H} \} &\in \underset{\forall k, \forall H \in \mc{H}^k }{\arg\max} \Pr\left( \mb{s} |\mb{\hat{x}},~H \right)                   											\label{eq-decouple-derivation-1}   \\
			       &= \underset{\forall H \in \mc{H}_{u} }{\arg\max} \Pr\left( \{ \mb{s}_{+}~\mb{s}_{z} \}|\mb{\hat{x}},~H \right) 											 \label{eq-decouple-derivation-2}  \\
			       &= \underset{\forall H \in \mc{H}^{+}_{u} }{\arg\max} \Pr\left( \mb{s}_{+}  |\mb{\hat{x}},~H \right) 		                										\label{eq-decouple-derivation-3}   \\
			       &= \underset{\forall A \in \mc{A}^{+}~ \forall H_i \in \mc{H}^{+}(A) }{\arg\max} \Pr\left( \mb{s}_{+}  |\mb{\hat{x}},~H_1 \hdots H_M \right)  					 \label{eq-decouple-derivation-4}  \\
			       &= \underset{\forall A \in \mc{A}^{+}~ \forall H_i \in \mc{H}^{+}(A) }{\arg\max} \prod_{i: A_i \in \mc{A}^{+} } \Pr\left( s_i | \textbf{child}(s_i),\mb{\hat{x}},~ H_i  \right) \label{eq-decouple-derivation-5}  \\
			       &= {\color{black} \bigwedge_{i: A_i \in \mc{A}^{+} } \underset{\forall H_i \in \mc{H}^{+}(A_i) }{\arg\max} \Pr\left( s_i | \textbf{child}(s_i),\mb{\hat{x}},~ H_i  \right) }    \label{eq-decouple-derivation-6}  		       
\end{align}
Here eq. \eqref{eq-decouple-derivation-1} is the original maximum likelihood hypothesis test over all $\mc{H}^{k}$ outages.
This is reduced to a search space over $\mc{H}_u$ in eq. \eqref{eq-decouple-derivation-2}.
This further reduces to an even smaller search space $\mc{H}^{+}_u$ due to processing of binary flow information in eq. \eqref{eq-decouple-derivation-3}. 
Eq. \eqref{eq-decouple-derivation-4} decouples $\mc{H}^{+}_u$ to a product space of local search hypotheses.
Eq. \eqref{eq-decouple-derivation-5} decouples the likelihood functions by conditioning on the set of observations.
This likelihood function is a product of terms which only depend on the local hypotheses.
Therefore maximizing this product is equivalent to maximizing each term separately, {\color{black}as in eq. \eqref{eq-decouple-derivation-6}}.
\begin{algorithm}[h]
\KwResult{Maximum Likelihood Hypothesis Detector}
\KwIn{
	[1] Load Forecast/Nominal Statistics: $\mb{\hat{x}}$, $\Sigma$ \\ \hspace{10.5mm}
	[2] Real Time Load: $\mb{s} = \{\mb{s}_{+}, \mb{s}_{z} \}$
}
$\mc{A}^{+} \leftarrow \textbf{prune}-\textbf{areas} ( \mc{A}, \mb{s})$ \\

\For { $ A_i \in \mc{A}^{+}$ } { 
	\textbf{// Generate Local Hypothesis set. } \\
	$\mc{H}^{+}_{u}(A_i) \leftarrow  \textbf{local}-\textbf{hypotheses}( \mc{T}, s_i, d(s_i) )$ \\

	\textbf{// Local MAP Detector} \\ 
	$\hat{H}_i \leftarrow \underset{\forall H_i \in \mc{H}^{+}(A_i) }{\arg\max} \Pr\left( s_i | d(s_i), \mb{\hat{x}},~H_i  \right)$ \\
}
{\color{black}
\textbf{// Combine Local Hypotheses} \\
 $\hat{H} \leftarrow \underset{ A_i \in \mc{A}^{+} } { \bigwedge } \hat{H}_i $ 
}
\caption{Maximum Likelihood Hypothesis Detector.}
\label{alg-decentralized-ML-detector}
\end{algorithm}
The decoupling of the centralized likelihood function in \eqref{eq-decouple-derivation-5}, leads to a simple decentralized detector in Algorithm \ref{alg-decentralized-ML-detector}.
The input is (1) the set of load forecasts $\hat{\mb{x}}$ with their nominal statistics $\Sigma$, (2) and the real time load information $\mb{s}$. 
The function \textbf{prune}-\textbf{areas} simply discards areas with zero flow.
The function  \textbf{local}-\textbf{hypotheses} performs the generation of local hypothesis set $\mc{H}^{+}_{u}(A_i)$ as described in section \ref{subsubsection-general-hypothesis-decoupling}.
   
Finally the local MAP detector is simple to evaluate as a multi-hypothesis test involving scalar gaussian of known means and variances as follows:
For a local area, since $\mc{H}^{+}(A)$ is enumerated, we determine:
\begin{align*}
\hat{H} = \underset{H \in \mc{H}^{+}(A)} {\arg\max}  \left( \frac{ s_i - \sum_{s \in \textbf{child}(s_i)} s - \mu(\mb{\hat{x}}, H) } { \sigma(\Sigma, H) }\right)^2 
\end{align*}
Therefore each decision only depends on an effective measurement 
\begin{align}
\Delta s_i \triangleq s_i - \sum_{s \in \textbf{child}(s_i)} s \label{eq:definition-effective-measurement}
\end{align}
 for each local area.
Computation of the means and variance are discussed in Section \ref{subsection-Evaluating-P_MAX-for-tree-networks}.
\section{Sensor Placement Problem}
\label{section-Placement-Algorithm}

In evaluating a placement, we use the maximum missed detection probability over all hypotheses.
This is a useful metric since we are considering a very large number of alternative hypotheses.
{\color{black} First consider the maximum error $P_{E}(H, \mc{M}) =  \Pr( \hat{H} \neq H ; \text{placement}~\mc{M} )$, where $\hat{H}$ is the optimal solution of the outage detection.}

This can be used for sensor placement evaluation.
\begin{align}
\mc{M}^{\star} = \underset{|\mc{M}|=M}{\arg\min}~ \max_{H \in H_u } P_{E}( H, \mc{M} ) \tag{OPT-1} \label{min-max-hyp-all}
\end{align}
This is difficult to do outside of a combinatorial enumeration of sensor locations and hypotheses.
A suitable proxy we optimize instead is the following:
\begin{align}
\mc{M}^{\star} = \underset{|\mc{M}|=M}{\arg\min}~\max_{A \in A } P_{E}(A ) \tag{OPT-2} \label{min-max-hyp-area}
\end{align}    
Where $P^{max}_{E}(A) = \max_{f, H \in H^+_u(A, f) } P_{E}( H, \mc{M} )$, that is we only search over the hypothesis in the local area.
This second optimization very closely approximates an upper bound to \eqref{min-max-hyp-all}  (see appendix \ref{subsection-proxy-function-optimization} for details). 
Optimization OPT-2 is solved via a bisection method on the following feasibility problem:
\begin{align}
\mc{M}^{\star} &= \underset{|\mc{M}| \leq M}{\text{find}}~\mc{M}  \tag{OPT-3} \label{feasibility-problem} \\
                       &~~~~\text{s.t.}~P_{E}^{\max}(A) \leq P^{target}                  \nonumber
\end{align}
So the minimum $P^{target}$ is determined which yields a solution of size $|\mc{M}| = M$.
This can be solved very efficiently, with the algorithm that follows.
\subsection{Feasibility Placement Algorithm}   
\label{subsection-Placement Algorithm-in-Line-Network}   
The intuition for the greedy placement algorithm is the following.
{\color{black}
Starting from the bottom of a tree, we successively maintain a temporary local area with root sensor in $e_t$.
The root sensor is iteratively moved closer to the root edge $e_1$, while maintaining that the maximum error of all areas is less than $P^{target}$.
This is done by maintaining that the local area has error less than $P^{target}$. 
If this is true we move closer to the root, if not, we place a sensor and start a new area.
Since the objective function decouples across areas, we can maintain that the feasibility problem is always satisfied.
Finally, if the number of sensors are less than $M$, $\mc{M}^{\star}$ is returned.
}
\begin{algorithm}[h]
\KwResult{ Placement for a Tree Network }     
\KwIn{
	[1] Tree network $\mathcal{T}$ \\ \hspace{10.5mm} 
	[2] Nominal loads statistics $\hat{\mb{x}}$, $\Sigma$ \\ \hspace{10.5mm}
	[3] Subproblem Ordering $V_{process}$ \\ \hspace{10.5mm}
	[4] Target error $P^{\text{target}}$
}
\textbf{// Generate node process ordering } \\
$E_{\text{process}} \leftarrow$ {\bf generate}-{\bf edge}-{\bf order}$(\mc{T})$ \label{tree-alt:generate-v-process}

\textbf{// initialize sensor placement as empty } \\
$\mc{M}^{g} \leftarrow \emptyset$ \\

\For { $ e_{t} \in E_{\text{process}}$ } {  \label{tree-alg:v-process-iteration}
	$A~\leftarrow~ $\textbf{construct}-\textbf{area}$( e_{t}, \mc{M}^{g}  )$ \label{tree-alg:construct-area-network} \\
	
	\textbf{// Evaluate the current subtree maximum missed }  \\
	\eIf { $P_{E}^{\max}(A )~\leq~P^{\text{target}}$ } { \label{tree-alg:eval-mmdp}
		\textbf{// continue to next node }
	} {  
		\uIf { $| \bf{child}( v_t )| == 1$  } {	
			$\mc{M}^{g}~\leftarrow~$\bf{line}-\bf{action}$(A, \mc{M}^{g})$
		} 
		\uElseIf{ $| \bf{child}( v_t )| > 1$  } {  
			$\mc{M}^{g}~\leftarrow~$\bf{tree}-\bf{action}$(A, \mc{M}^{g})$
		}
	}
}
\Return $\mc{M}^{g}$
\caption{Solution to optimization \eqref{feasibility-problem} for tree network.}
\label{tree-placement-algorithm}
\end{algorithm}
This framework can be realized in Algorithm \ref{tree-placement-algorithm}.
The inputs to the method are the tree network $\mc{T}$ and the set of nominal load forecasts $\hat{\mb{x}}$ and forecast variance $\Sigma$. 
{\color{black}
To have the effect of starting at the leaf of the network and move our way up to the root, the algorithm will process a sequence of edges $E_{\text{process} }$.
For example in Figure \ref{fig:example_tree_10_node}, $E_{\text{process} } = \{E(b_4)~E(b_5)~E(b_3)~E(b_2)~E(b_1) \}$.
We generate the list $E_{\text{process} }$ in line \ref{tree-alt:generate-v-process} with function {\bf generate}-{\bf edge}-{\bf order}.
}
The function takes the tree $\mc{T}$ and traverses via breadth first search keeping track of the depth of each vertex/edge.
Reversing this list of depths yields a list of nodes to process, {\color{black} the parent edge being $e \in E_{\text{process}}$.}

The greedy solution $\mc{M}^{g}$ must first be initialized as empty.
In line \ref{tree-alg:v-process-iteration} we iterate over the current root node $e_t$ and current sensor placement $\mc{M}$.
In line \ref{tree-alg:construct-area-network} we construct the current area network $A$.
We then evaluate $P_{E}^{\max} (A)$ in line \ref{tree-alg:eval-mmdp}; 
If $P_{E}^{\max} (A) > P^{\text{target}}$ we perform a placement action {\bf line}-{\bf action} or {\bf tree}-{\bf action} depending on the number of child nodes of $v_t$ (downstream of $e_t$).
Each sub function is described in more detail as follows.

\subsubsection{{\bf construct}-{\bf area}}
For each iteration, the temporary node $e_t$ is visited and the area network $A$ is constructed with $e_t$ and the previous solution $\mc{M}^{g}$.
{\color{black} The current edge $e_t$ is the temporary root sensors of the area $s_t$.
The terminal sensors of the area are any sensors in $\mc{M}^{g}$ that are children of $s_t$.
Note that this may be empty at the start of the algorithm.
}
\subsubsection{{\bf line}-{\bf action}}
Given that our current subproblem satisfies $P_{E}^{\max}(A) < P^{\text{target}}$, if the next subproblem does not satisfy the condition, our only option is place a sensor at {\bf child} $(v_t)$.
So $\mc{M}^{g} \leftarrow \mc{M}^{g} \cup e_{t}$.
\subsubsection{{\bf tree}-{\bf action}}
\label{subsubsection-tree-action}
\begin{figure}     
\centering
\includegraphics[width=0.5\textwidth, height=0.23\textwidth]{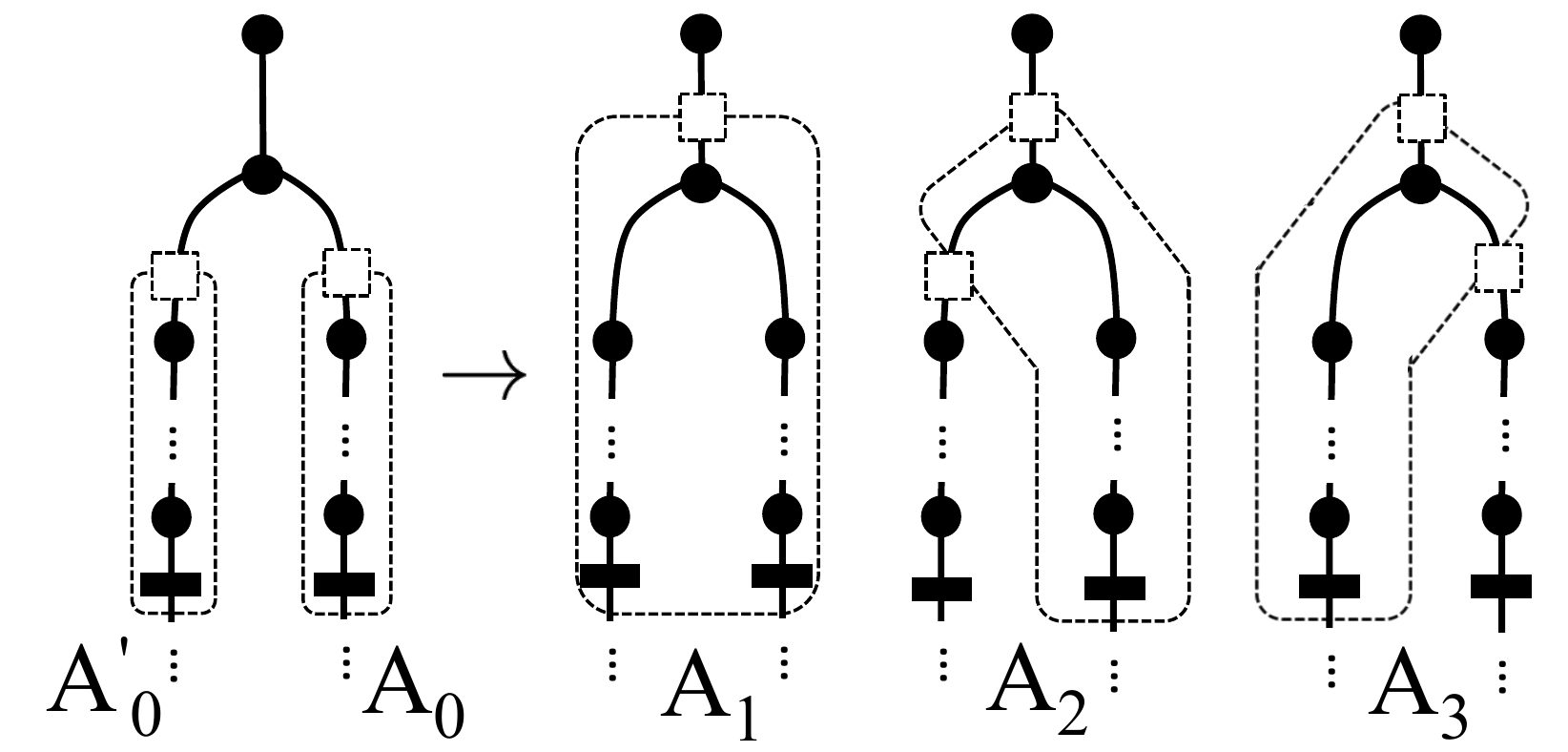}  
\caption{({\bf greedy}-{\bf tree}-{\bf action})
Correct node traversal in the tree network assumes that $P^{\max}_{E}(A_0) < P^{\text{target}}$ and  $P^{\max}_{E} (A^{\prime}_0) < P^{\text{target}}$.
If $P^{\max}_{E} (A_1) < P^{\text{target}}$ we do nothing.
Else, we generate and evaluate the error one of the $2^{| \bf{child}(v_t)|-1}$ area networks that can be constructed for example $A^{\prime}_1$ and $A^{\prime\prime}_1$.
}
\label{fig:tree_action}
\end{figure}
Given that the algorithm up to now has placed sensors on the two disjoint trees with roots with $v_1$ and $v_2$
We must move as far up to the root as possible before we are forced to place a measurement.
This leads to two different types of actions: a greedy strategy that is easy to implement and the optimal strategy.
The greedy strategy is implemented in practice and is almost always equal to the optimal strategy, which is discussed in Appendix \ref{appendix-greedy-optimal-tree-discussion}. 

The greedy strategy chooses the area network with the smallest $P_{E}^{\max} (A)$.
For example, in Figure \ref{fig:tree_action}, first assume $P_{E}^{\max} (A_0)$ and $P_{E}^{\max} (A^{\prime}_0)  < P^{\text{target}}$ and $P_{E}^{\max}(A_1) > P^{\text{target}}$.
In moving $e_t$ closer to the root, we must choose either $A_2$ or $A_3$ based on the placement which has the smallest error.
\subsection{Evaluating $P^{\max}_{E}(A)$}
\label{subsection-Evaluating-P_MAX-for-tree-networks}
{\color{black} 
The objective $P_{E}^{\max}(A_i)$ is computed as follows.
For each local hypothesis $H \in \mc{H}^{+}(A_i, f)$, the conditional distribution $\Delta s |~\hat{\mb{x}}~H \sim N(\mu(\hat{x}, H)$, $\Sigma(\hat{x}, H))$ is computed.

To compute this, we introduce the following $W^{\mu}(k) = \sum_{v \in \textbf{desc}(v)} \hat{x}(v)$ which computes the cumulative load forecast of all descendent vertices.
Similarly $W^{\sigma}(v) = \sum_{v \in \textbf{desc}(v)} \sigma^2(v)$ for the forecast variances.
For a given area, we evaluate the $\Delta s_i$ under hypothesis $H$, assuming the root $s_i$ and child sensor locations $\textbf{child}(s_i)$ for a particular observed flow.
\begin{align}
\mu(\hat{\mb{x}}, H) = W^{\mu}(s_i) - \sum_{ e \in V(H) } W^{\mu}(e) - \sum_{ e \in \textbf{child}(s_i) } W^{\mu}(e) 
\end{align}
Note that this is computed for a particular binary flow pattern.
Therefore, no outage edge upstream from the child sensors that we consider.
Likewise $\sigma(\hat{\mb{x}}, H)$ is computed via $W^{\sigma}$.

Finally, given the scalar distribution for $\Delta s_i | \{\hat{\mb{x}}~H \}$, the probability of missed detection for a scalar maximum likelihood detector can be computed.
}
\subsection{Optimality and Complexity}  
\label{subsection-Optimality-and-Complexity-of-Tree-Algorithm}  
This discussion concerns only {\bf optimal}-{\bf tree}-{\bf action} since it guarantees optimality, although the greedy strategy output is almost always identical.
\begin{thm}
\label{tree-network-optimality}
The bottom up placement solution $\mc{M}^{g}$ relying on {\bf optimal}-{\bf tree}-{\bf action} traversal solves OPT-3.
\end{thm}

Complexity of evaluating the detector and the objective as well as the placement algorithm is discussed in Appendix \ref{subsection-complexity-analysis}.
The main results are summarized as follows:
\begin{itemize}
\item The worst case complexity of evaluating the detector for an area and it's missed detection $P^{max}_{E}(A)$ for an area of size $|E|$ is $O(4^{|E|})$.
{\color{black} Evaluating only outages of size $k$ is $O(|E|^{k})$.}
\item Given a fixed size to evaluating $P^{max}_{E}(A)$ the greedy placement algorithm is of $O(|E|)$ complexity while the optimal strategy is of $O(4^{\log |E|})$ complexity.
\end{itemize}

{\color{black}    
The detector and placement complexity in the worst case is quite poor.
However in any practice, this cost is averted using a detector fixed hypothesis size.
The outage model was of each edge having some finite prior likelihood of outage.
Therefore a multi-edge outage of large size is much less likely.
For this reason, all k outages do not need to be enumerated.
In practical a single edge outage per area may be sufficient.
Using a single outage detector with greedy placement is extremely efficient (only $O(|E|^2)$ complexity).
}

\section{Distribution System Case Study}
\label{subsection-Distribution-System-Case-Study}
\textit{For the remaining case studies only single area outages are considered.}

\subsection{PNNL Case Study}
\label{subsection-PNNL-Case-Study}

\begin{figure}[h]
\hspace{-5mm} 
\includegraphics[width=0.5\textwidth, height=0.4\textwidth]{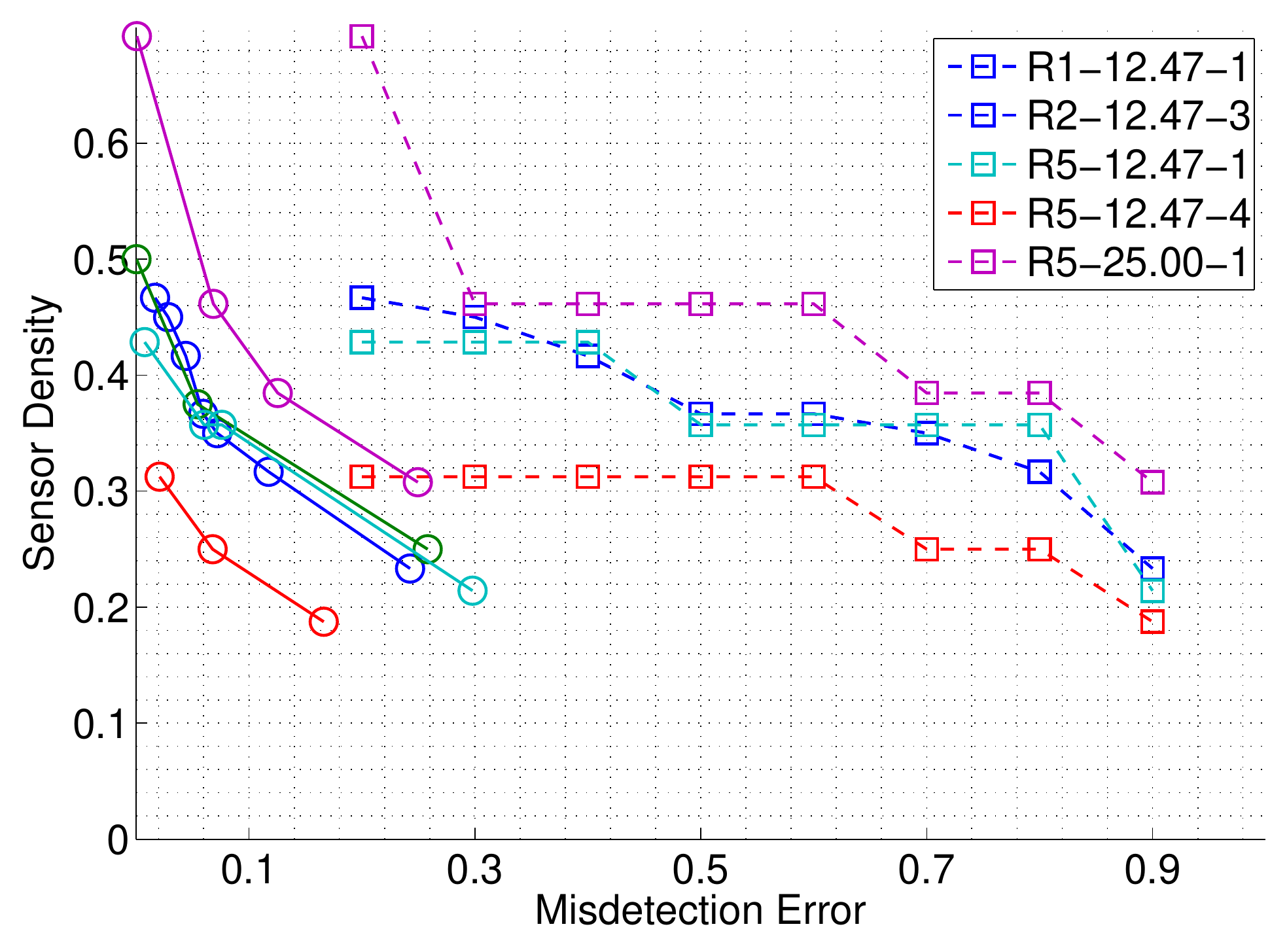}
\caption{
Outage detection performance for selected PNNL feeders.  
Square marker denotes specified error target for optimization.  
Circular marker denotes empirical mean missed detection error for hypothesis set.
}
\label{fig:pnnl_case_study}
\end{figure}
We perform outage detection using a subset of the Pacific Northwest National Laboratory test feeders \cite{PNNL2008}.
Table \ref{tab:pnnl_chosen_feeders} gives overview of the feeders chosen for the simulation study.
The primary applications of the feeders, are heavy to light urban networks, as well as suburban and rural networks.
The climate zones refer to (1) temperate (2) hot/arid (3) cold (4) hot/cold (5) hot/humid according to \cite{PNNL2008}.

\begin{table}[h]
\centering
\caption{PNNL Test feeders used in case study}
\begin{tabular}{@{}ccccccc@{}}
\toprule
Network       &  Voltage & Climate Zone  &  Type & Size      \\
\cmidrule{1-5} 
R1-12.47-1  &  12.5 kV      &   1   &   suburban      &  613    \\ 
R2-12.47-3  &  12.47 kV    &   2   &   urban            &   52     \\
R5-12.47-1  &  13.8 kV      &   5   &   urban            &  265    \\
R5-12.47-4  &  12.47 kV    &   5   &   commercial   &  643    \\
R5-25.00-1  &  22.9 kV      &   5   &   suburban      &  946    \\
\bottomrule 
\label{tab:pnnl_chosen_feeders}
\end{tabular}  
\end{table}

\subsubsection{Outage Model}
For the PNNL feeders, outages are simulated by \text{fuses} and \textit{switches} disconnected the downstream loads from the substation feeders. 
For each network, all fuses and switches are reduced to edges in the general tree network.
The set of loads which are disconnected by a fuse or switch disconnecting are lumped to aggregate loads.
For these loads, the mean load of each group of fuses can vary. 

\subsubsection{Forecast Error Model}
{\color{black} 
In \cite{Sevlian2014C} the authors present a rule of thumb model for day ahead load forecasting at various aggregation levels based on smart meter data.
The day ahead forecast coefficient of variation, $\kappa = \sigma/\mu$ is shown to be dependent on the mean load of the group.
Many studies make simplified assumptions on the relative forecast error.
However, at the level of small aggregates, the forecast $\kappa$ can vary greatly on the size of the aggregate and must be taken into account.
A Reasonable fit shown in \cite{Sevlian2014C} is $\kappa(W) =  \sqrt{ \frac{3562 }{W} + 41.9}$.
This formula is used to show that each set of islanded loads will have a different value of $\kappa(W)$.
}

Figure \ref{fig:pnnl_case_study} shows the application of the sensor placement algorithm for each network.
Even though each networks represents a different applications, they show somewhat similar performance in terms of placement density. 
Averaged over each of the network configurations, attaining $10 \%$ mean missed detection error is possible by having $30 \%$ sensor density.
Seen another way, we can reduce the realtime monitoring of each fuse by $70 \%$ by tolerating a small amount of error in the outage decision.

\subsection{General Line and Tree Network Sensitivity}
\label{subsection-Sensor-Placement-Sensitivity-Analysis}

\begin{figure}[h]
\hspace{-5mm}
\subfigure[][]{
	\label{fig:part1_p1_line}
	\includegraphics[width=0.23\textwidth, height=0.23\textwidth]{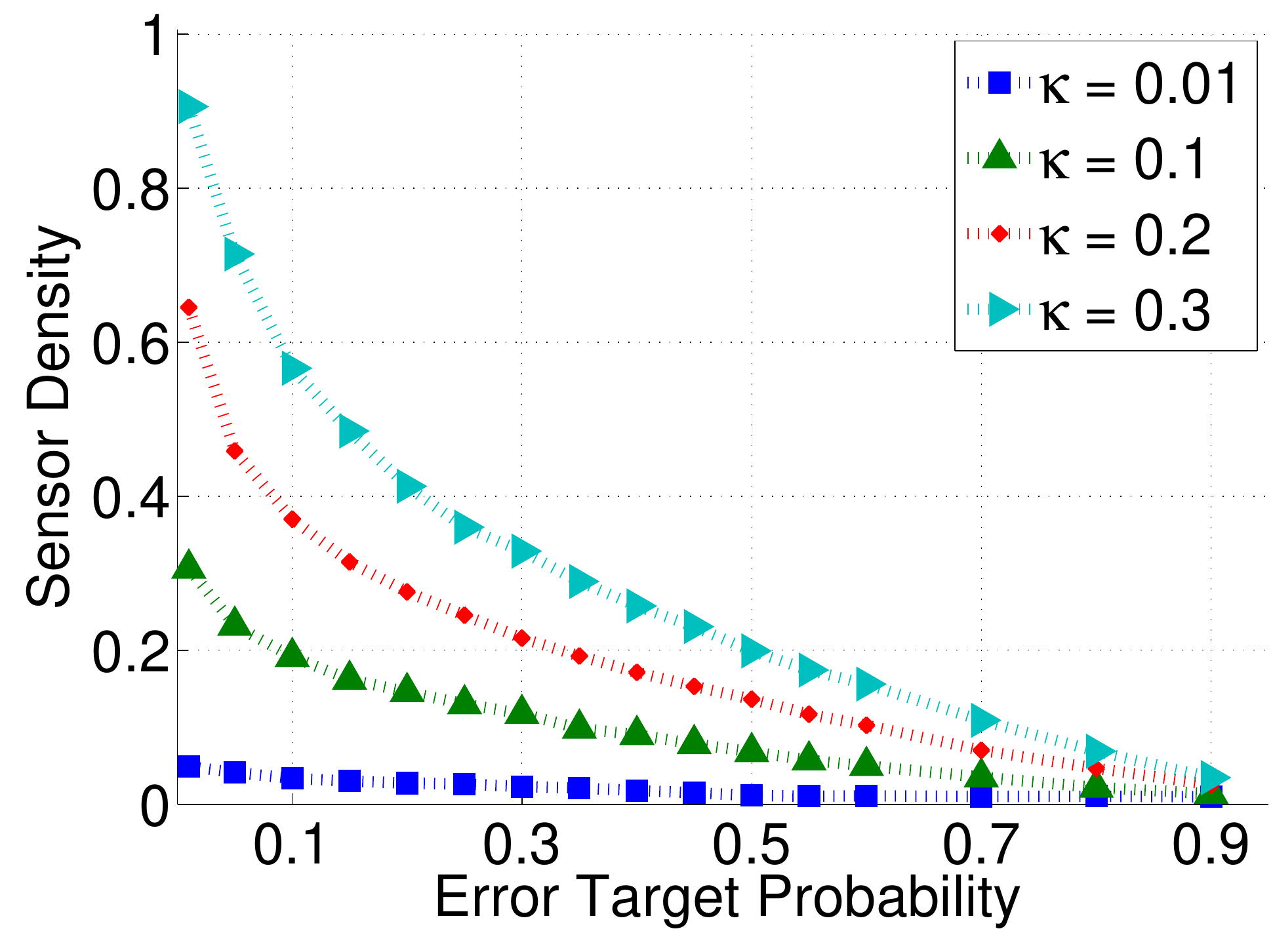}
}
\subfigure[][]{
	\label{fig:part1_p1_tree}
	\includegraphics[width=0.23\textwidth, height=0.23\textwidth]{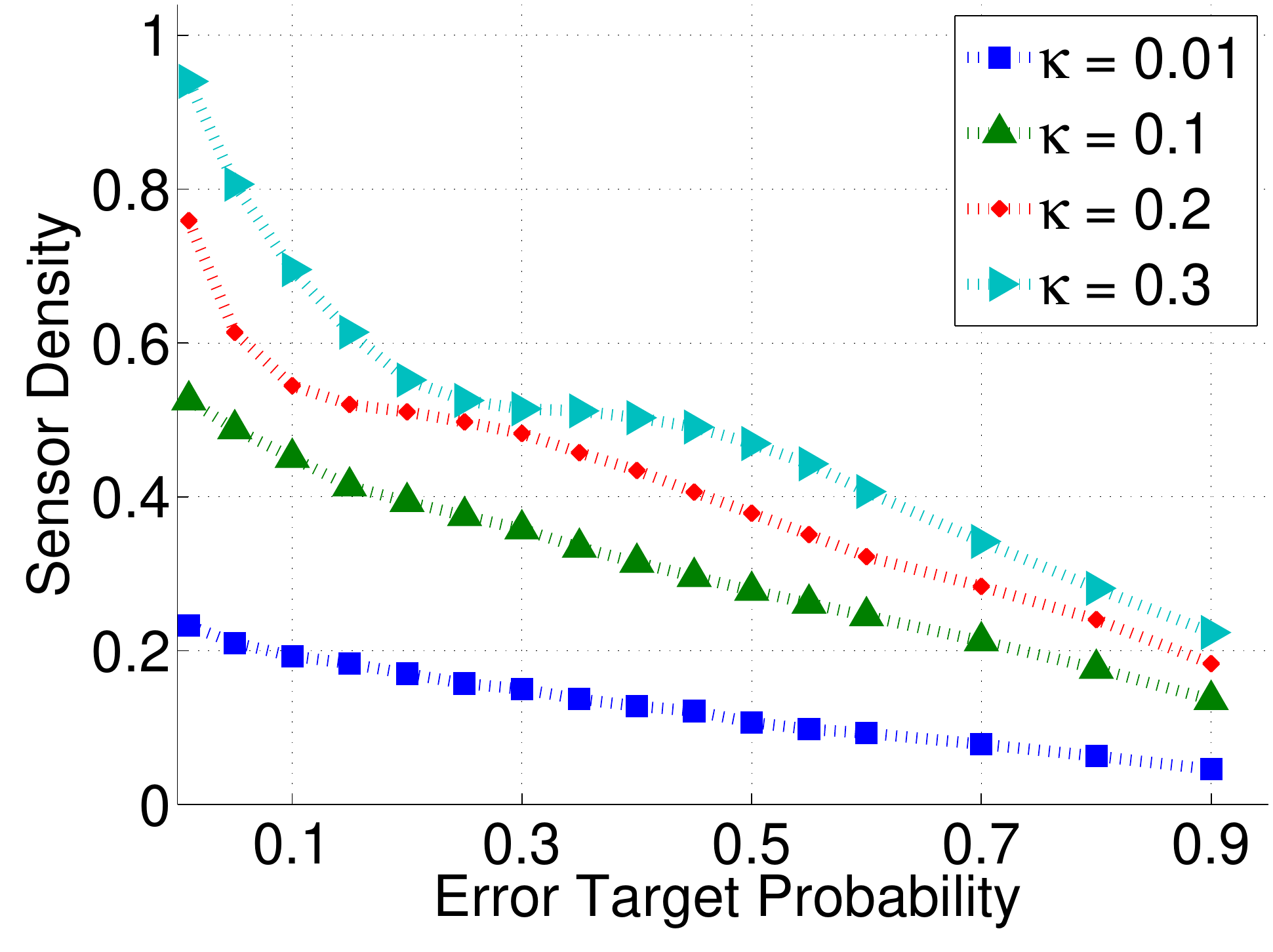}
	
}
\caption{ 
The effect of optimization error target $P^{\text{target}}$ and relative forecast error $\kappa$ on both line (\ref{fig:part1_p1_line}) and tree (\ref{fig:part1_p1_tree}) networks. 
}
\label{fig:part1_plots}
\end{figure}

Figure \ref{fig:part1_plots} shows the sensitivity of the line and tree networks under different simulation parameters.
Both networks are of length 100 nodes, the tree was generated using the method in \cite{Rodionov2003}.
In an ideal line network with extremely high forecast accuracy ($\kappa = 1 \%$), 1 or 2 sensors are required for extremely low missed detection errors.
This extreme situation does not occur in practice, but serves as a baseline for realistic networks.
From Figure \ref{fig:part1_p1_line} we see that the required sensor density decreases quite quickly vs. $P^{\text{target}}$.
The relation between sensor density and $P^{\text{target}}$ is smoothly decaying.
In comparison, randomly generated tree networks require on average $2-3$ times as many sensors to achieve the same error target.
  
\subsection{Missed Detection Error}
\label{subsection-Missed-Detection-Error}

Optimization \eqref{min-max-hyp-area} is meant to minimize the maximum missed detection error among all possible hypothesis.  
This is clearly can be too conservative of a requirement.
Therefore it is useful to understand the nature of the actual hypothesis missed detection values that arise from a given sensor placement.
\begin{figure}[h]
\hspace{-5mm}
\subfigure[][]{
	\label{fig:missed_detection_histogram}
	\includegraphics[width=0.24\textwidth, height=0.23\textwidth]{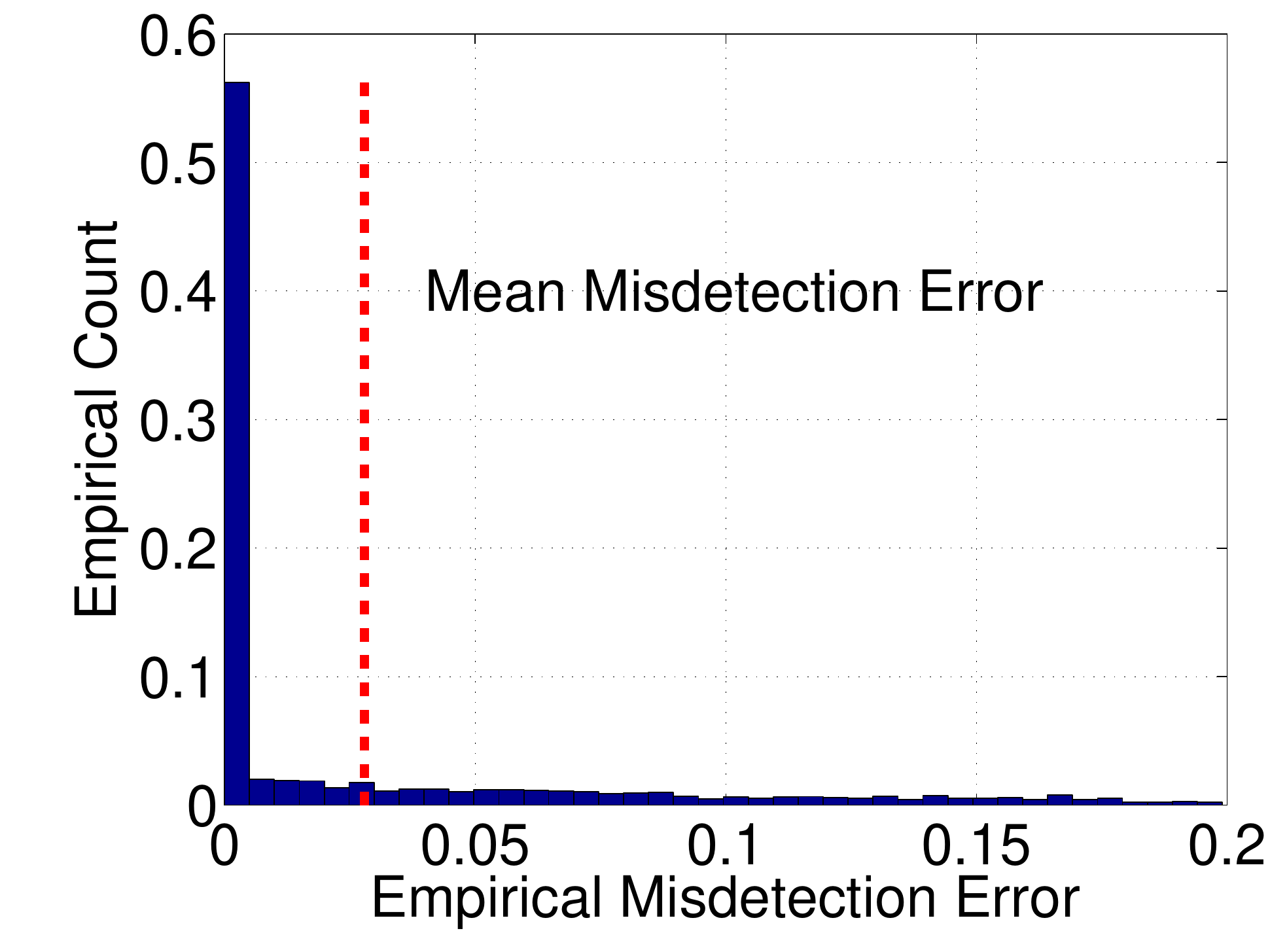}
}
\subfigure[][] {
	\label{fig:sd_vs_et_mean_max}
	\includegraphics[width=0.24\textwidth, height=0.22\textwidth]{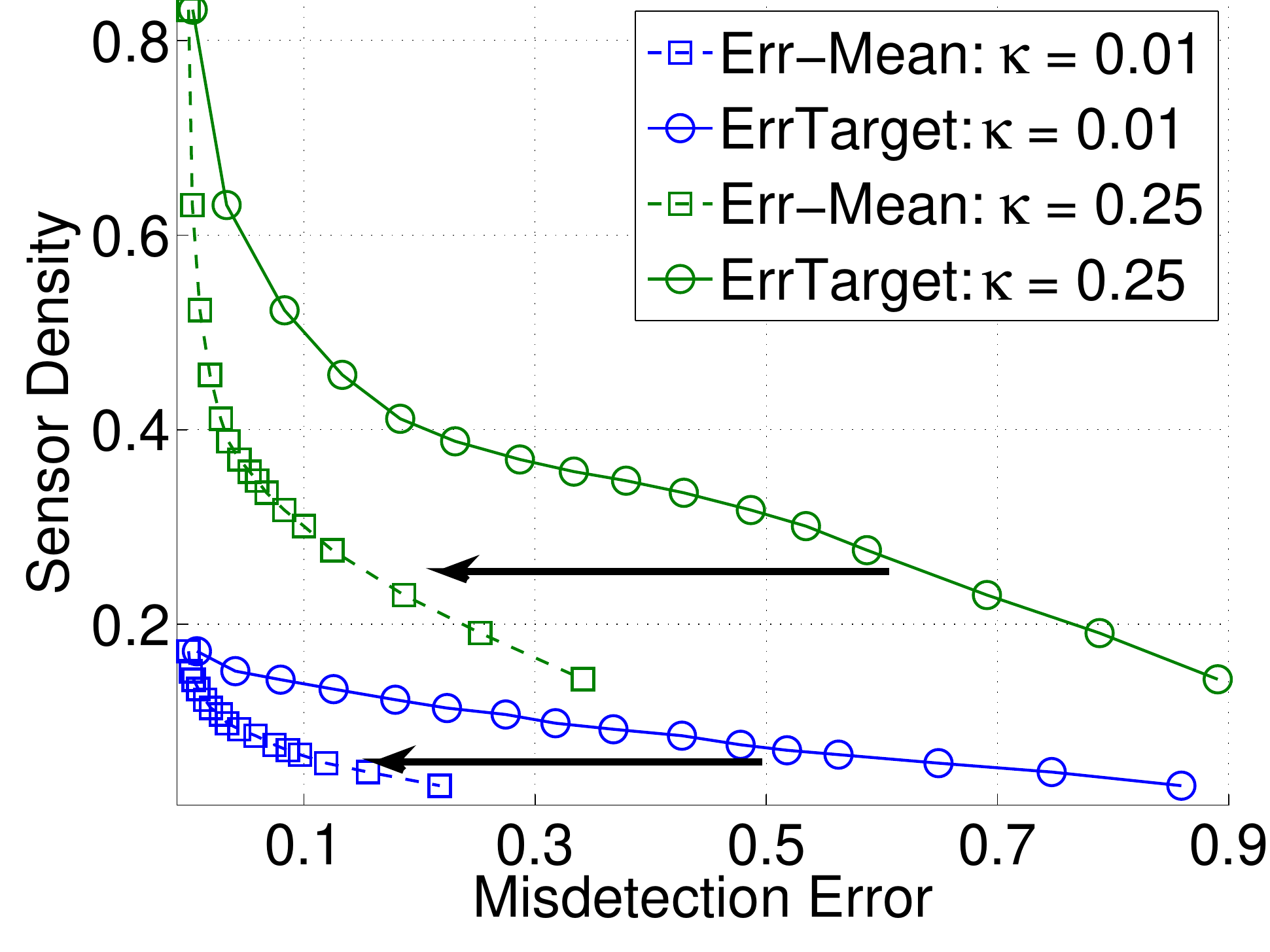}
}
\caption{Hypothesis missed detection analysis: \ref{fig:missed_detection_histogram} error histogram for $P^{\text{target}} = 0.2$; \ref{fig:sd_vs_et_mean_max} tree network reduction to mean error.
}     
\label{fig:mean_error_plot_1}
\end{figure}

Figure \ref{fig:missed_detection_histogram} shows the distribution of missed detection probabilities for a tree network.
Setting $P^{\text{target}} = 0.2$, and $\kappa = 0.3$ we record the value of each hypothesis error.
The empirical maximum error is close to the target 0.2.  
This makes sense because in successively solving the feasibility problem, we will expand the area network until the maximum error surpasses $P^{\text{target}}$.
However, we see that in fact almost all of the missed detection probabilities are less than the target.
For this example in particular $34 \%$ of the hypothesis are less than $1e^{-3}$ therefore essentially zero.

In comparing the mean and maximum errors for the range of achievable values of $\kappa$ and $P^{\text{target}}$.
The mean error is in on average $25 \%$ of the $P^{\text{target}}$.
The maximum error in the network and $P^{\text{target}}$ very closely, therefore the optimization yields a very tight result.

\section{Conclusion}
We propose an outage detection framework combining power flow measurements on edges of the distribution system along with consumption forecasts at nodes of the network.
We formulate the detection problem and provide an optimal placement for the maximum missed detection error metric.
Finally, relying on feeder information from the pacific northwest national labs as well as a forecast error scaling law derived from Pacific Gas and Electric smart meter data, we demonstrate our formulation.
\label{section-Conclusion}    
\bibliographystyle{plain}  
\bibliography{MAIN_BIB}  
\newpage
\appendix
\subsection{Nomenclature Table}
\addcontentsline{toc}{}{Nomenclature}
\begin{IEEEdescription}[\IEEEusemathlabelsep\IEEEsetlabelwidth{$V_1,V_2,V_3, V_4$}]
\item[$\mathcal{T}$]                                            Tree network representation of a distribution feeder
\item[$(V, E)$]                                                     Vertex and edge set of tree $\mathcal{T}$
\item[$V(H)$]                                                       Set of vertices that are connected to root under outage hypothesis $H$.
{\color{black}
\item[$\textbf{desc}(v)$]                                      Descendants of vertex $v$.
\item[$\textbf{child}(v)$]                                       Children of vertex $v$.
}
\item[$x_n$, $\mb{x}$]                                         Scalar load and forecast value at vertex $v$
\item[$\epsilon_n$, $\mb{\epsilon}$]                    Forecast residual for load $l(v)$
\item[$\sigma_{n}^2$, $\Sigma$]                         Forecast residual variance and covariance matrix.
\item[$\mc{H}^{1}$]                                              Single outage hypothesis and element and set.
\item[$\mc{H}^{k}$]                                              k-outage hypothesis. 
\item[$\mc{H}_{u}$]                                              Unique hypothesis where no edges are downstream of any others.
\item[$A$, $\mc{A}$, $\mc{A}^{+}$]                     Local area network $A \in \mc{A}$; Pruned areas under binary flow processing.	        
\item[$\mc{H}^{+}_{u}(A_k, f)$]        			Set of local area hypotheses post binary flow processing.
{\color{black}
\item[$\mc{H}_{k, i(k)}$]           			        $i(k)^{th}$ hypothesis in area $A_k$.  Used to reconstruct global hypothesis $H_i$.
}
\item[$\mc{M}$]                                                  Sensor placement $(\mathcal{M} \subset E)$
\item[$\mathbf{s}$]                                             Set of observations on edges of network
\item[$r_{ij}$]          					        Acceptance region for pairwise test of hypotheses: $H_i$ and $H_j$
\item[$R_{\mc{H}}(H)$] 		                          Acceptance region of hypothesis $H$ over all alternatives in $\mc{H}$.
\item[$P_{E, C}(H, \mc{M})$, ]                            Probability of error (E) and correct detection (C) for hypothesis $H$, under placement $\mc{M}$.
\item[$P_{E}^{max}(A)$]                                     Maximum probability of incorrect detection over all hypothesis error in area A.
\item[$P_{C}^{min}(A)$]                                      Minimum probability of correct detection over all hypothesis error in area A. 
\end{IEEEdescription}

\subsection{MAP Detection for Outage Hyptheses}
\label{subsection-MAP-Detection-for-Outage-Hyptheses}
Here we show how the general MAP detector rule can be evaluated for where we combine edge flows $\mb{s}$, load forecasts $\hat{\mb{x}}$ and candidate outages $H$.
\begin{align}
\hat{H} &= \underset{ H \in \mc{H}^{k} }{\arg\max} \Pr\left( H~|~\mb{s},~\mb{\hat{x}} \right)                                                                                                         \label{apx-map-l1} \\
 		   &=  \underset{ H \in \mc{H}^{k}  }{\arg\max} \frac{ \Pr\left( \mb{s},~\mb{\hat{x}} ~ | ~ H \right) \Pr\left( H \right) }{ \Pr\left( \mb{s},~\mb{\hat{x}} \right)  } \label{apx-map-l2} \\
    		   &=  \underset{ H \in \mc{H}^{k} }{\arg\max} \Pr\left( \mb{s},~\mb{\hat{x}} ~ | ~ H \right) \Pr\left( H\right)                                                                           \label{apx-map-l3} \\
	           &= \underset{  H \in \mc{H}^{k} }{\arg\max} \Pr\left( \mb{s}~ | ~\mb{\hat{x}},~ H\right)  \Pr\left( \mb{\hat{x}}~|~ H \right) \Pr\left(H \right)                         \label{apx-map-l4} \\ 
	           &= \underset{  H \in \mc{H}^{k} }{\arg\max} \Pr\left( \mb{s}~ | ~\mb{\hat{x}},~ H \right)  \Pr\left( \mb{\hat{x}} \right) \Pr\left( H \right)                                \label{apx-map-l5} \\ 
                    & = \underset{  H \in \mc{H}^{k} }{\arg\max} \Pr\left( \mb{s}~ |~\mb{\hat{x}},~H \right)  \Pr\left( H \right)                                                                           \label{apx-map-l6} \\
                    & = \underset{  H \in \mc{H}^{k}  }{\arg\max} \Pr\left( \mb{s}~ |~\mb{\hat{x}},~H \right)                                                                                                     \label{apx-map-l7}    
\end{align}

Lines \eqref{apx-map-l1} - \eqref{apx-map-l3} convert the MAP detector to a likelihood detector with prior weights.
Line \eqref{apx-map-l4} conditions on the load forecast $\mb{\hat{x}}$.
Since $\mb{\hat{x}}$ does not depend on the outage hypothesis, (only $\mb{s}$ does), the term can be removed leading to \eqref{apx-map-l6}.
In \eqref{apx-map-l7}, we assume a uniform prior over all hypotheses, however this does not have to be the case.

Given the assumption of each edge going into outage with some fixed prior probability, a single edge outage hypothesis should have $\Pr\left( H \right)  = \rho$, while a k outage condition should have a prior of $\Pr\left( H \right)  = \rho^k$.
This motivates enumerating fewer outage hypotheses when evaluating $\mc{H}_u$ in practice.

\subsection{Extended Discussion of Recursive Evaluation of $\mc{H}_u$}
\label{subsection-Extended-Discussion-of-Recursive-Evaluation-of-H_u}
\begin{figure}[h]
\centering
\subfigure[][]{
\includegraphics[scale=0.25]{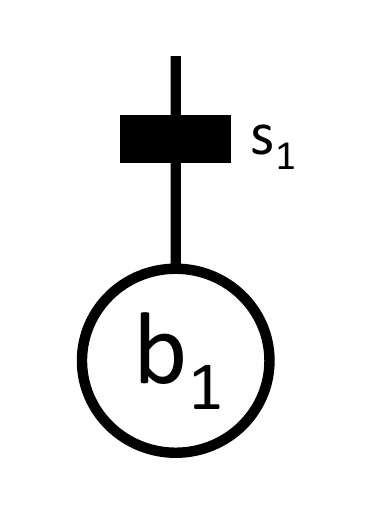}  
\label{fig:Hu_example_1}
}
\subfigure[][]{
\includegraphics[scale=0.25]{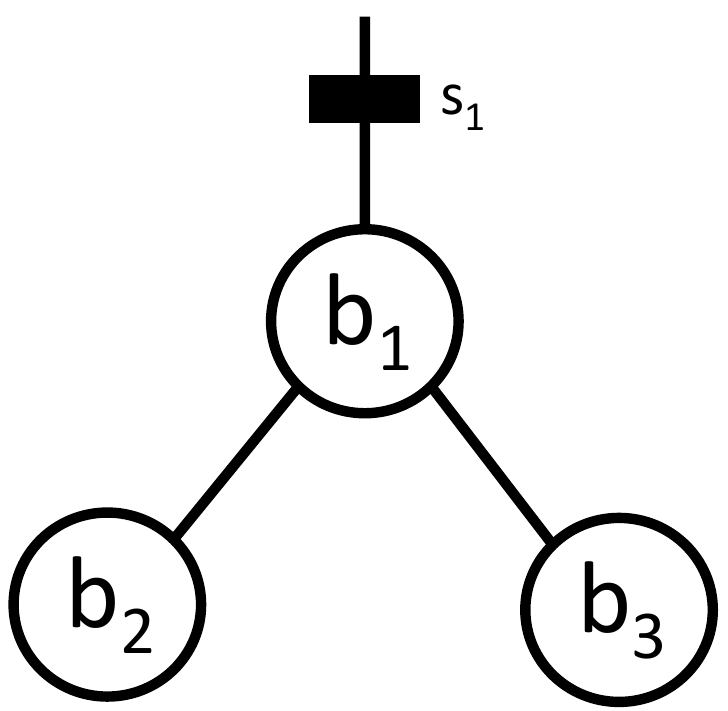}  
\label{fig:Hu_example_2}
}
\subfigure[][]{
\includegraphics[scale=0.25]{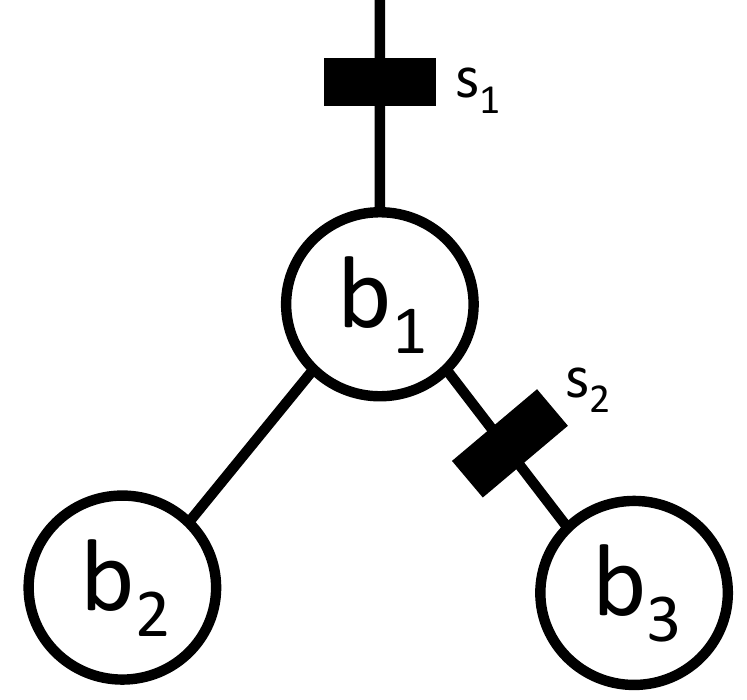}  
\label{fig:Hu_example_3}
}
\caption{   
\ref{fig:Hu_example_1} Worked example 1.
\ref{fig:Hu_example_2} Worked example 2.
\ref{fig:Hu_example_3} Worked example 3.
}
\end{figure}

We present here additional worked out examples, that show how the recursive definition can enumerate all hypotheses and focus on some corner cases that must be defined.
Recall the recursive definition:
\begin{align*}
\mc{H}_u(b) = E(b) \cup \left(  \bigcup_{ \mb{b} \in \mbb{P}( \textbf{child}(b) ) }   \left(  \underset{b \in \mb{b}} { \bigwedge } \mc{H}_{u}(b)  \right)  \right) 
\end{align*}

\subsubsection{Example 1 (Figure \ref{fig:Hu_example_1})}
This is the simplest case to evaluate and is:
\begin{align*}
\mc{H}_u &= E(b_1) \cup \emptyset
\end{align*}
The null hypothesis set arises from evaluating $\mbb{P}( d(b_1) ) = \emptyset$, since $b_1$ has no children.

\subsubsection{Example 2 (Figure \ref{fig:Hu_example_2})}
This is the simplest case to evaluate and is:
\begin{align*}
\mc{H}_u &= \{ E(b_1) \cup \emptyset \cup \mc{H}_u(b_2) \cup \mc{H}_u(b_3) \cup \mc{H}_u(b_2) \times \mc{H}_u(b_3) \} \\ 
	        &= \{ E(b_1) \cup \emptyset \cup  \{ \emptyset \cup E(b_2) \} \cup \{ \emptyset \cup E(b_3) \} \cup \\
	        &~~~~~~~ \{ \emptyset \cup E(b_2) \} \times \{ \emptyset \cup E(b_3) \} \} \\ 
	        &=  \{ \emptyset \cup E(b_1) \cup E(b_2) \cup E(b_3)  \cup \left( E(b_2) \times E(b_3) \right) \} 
\end{align*}
This example is reduced to it's final form in eq. \eqref{eq:HuB3}, in Section \ref{subsection-unique-outages}.
However, the following equalities are omitted in the enumeration:
\begin{align}
\emptyset \cup \emptyset  &= \emptyset \\
\emptyset \times \emptyset &= \emptyset \\ 
\emptyset \times e_i          &=  e_i 
\end{align}
The final relation leads to $\emptyset \times E(b) = E(b)$.   

\subsubsection{Example 3 (Figure \ref{fig:Hu_example_3})}
Consider the two binary flow indicators $\mbb{I}_{ \{ s_1, s_2 > 0 \} }=\{1~1\}$ and $\{1~0\}$.
In the first case, we have the following product set:
\begin{align}
\mc{H}^{+}(\{1~1\}) &= \{ E(b_2) \cup \emptyset_1 \} \times \{ E(b_3) \cup \emptyset_2 \}  \nonumber \\
	          &= \{ E(b_2) \times E(b_3) \cup \emptyset_1 \times E(b_3) \cup E(b_2) \times \emptyset_2 \cup \emptyset_2  \times \emptyset_2 \}  \nonumber \\ 
	          &= \{ E(b_2) \times E(b_3) \cup E(b_2) \cup E(b_3)  \cup \emptyset \}  \nonumber
\end{align}

We use the fact that $\emptyset \times E(b) = E(b)$ and that the product $\emptyset_1 \times  \emptyset_2 = \emptyset$, which is the global null hypothesis from the naive enumeration in Example 2.
Similarly, enumerating the $\{1~0\}$ case, we have that: $\mc{H}^{+}(\{1~0\}) = \{ E(b_1) \}$.
We see that splitting the hypotheses based on flow information, conserves the search space, since $\mc{H}_u^{+}(\{1~1\})  \cup  \mc{H}_u^{+}(\{1~0\}) =  \mc{H}_u$.

\subsubsection{Tree $\mc{T}_1$}

For tree $\mc{T}_1$, we have:
\begin{align*}
\mc{H}_{u}(b_1) &= \{ \emptyset \cup E(b_1) \cup E(b_2) \cup \mc{H}_{u}(b_3) \cup E(b_2) \times \mc{H}_{u}(b_3)  \} \\ 
\mc{H}_{u}(b_3) &=  \{ \emptyset \cup E(b_3) \cup E(b_4) \cup  E(b_5)  \cup \{ E(b_4) \times E(b_5) \} \}
\end{align*}

\subsection{General Hypothesis Decoupling}
\label{subsubsection-general-hypothesis-decoupling}

These two cases provide the intuition for a general procedure which is as follows:
Given binary information from flows, all areas $A_k$ with rooted sensor with $s_k = 0$, are discarded in generating a local hypothesis.
Each node in the branch graph is assigned a label, $l \in L$  $L = \{P, Z, U \}$ for ($P$) positive, ($Z$) zero, and ($U$) undetermined branches.
These are defined as follows:
\begin{enumerate}
\item [] \textbf{Positive Branch}: Branch is upstream from a sensor measuring positive flow, therefore can never be evaluated in any outage hypothesis.
Also, it's immediate parent branch cannot be enumerated either.
\item [] \textbf{Zero Branch}: This branch is directly upstream from a zero measurement therefore it's edges must always be enumerated in any outage hypothesis.
\item [] \textbf{Undetermined Branch}: This branch has no information, so is enumerated without any restriction.
\end{enumerate}

This definition leads to the following procedure to enumerate $\mc{H}^{+}(A_i, \mbb{I}_{\{\textbf{child}(s_i) > 0\} } )$. 
First each branch-node is labeled with the following procedure:
\begin{enumerate}
\item [] \textbf{Initialization} Branch with descendent sensor (1) $s > 0$ assigned label $P$ and (2) $s = 0$ assigned label $Z$, and (3) no descendants assigned label $U$.
\item [] \textbf{Update} Given a current branch node $b$ and the set of children, the node is assigned as follows:
(1) If any descendent node is labelled $P$ then it must be labelled $P$.  (2) If descendants are $U$ and $Z$, then it must be labeled $U$.
\end{enumerate}

Once the the branch-nodes are labeled, enumerating $\mc{H}^{+}(A_i, \mbb{I}_{\{ \textbf{child}(s_i) > 0\} } )$ can be done recursively using the following rules:
\begin{enumerate}
\item [] \textbf{Positive  Rule} Never enumerate a branch $(E(b))$ with positive flow label ($P$).
\item [] \textbf{Zero Rule} When evaluating the recursive definition $\mc{H}(b)$ on an element of the power set.
If any descendent is labelled $Z$, only evaluate product set elements that contain this branch.
\end{enumerate}
   
\begin{figure}[h]
\centering
\subfigure[][]{
\includegraphics[scale=0.30]{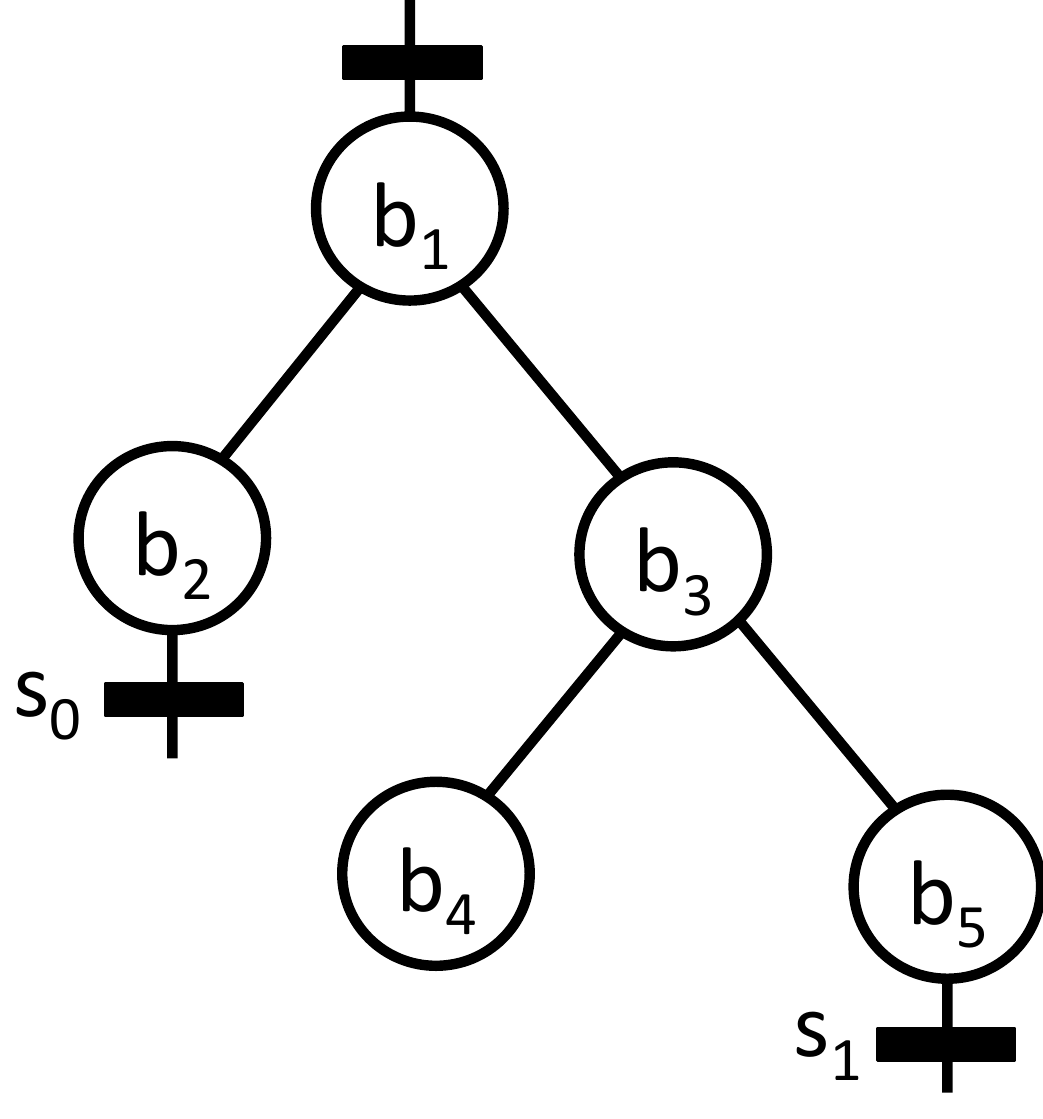}  
\label{fig:local_area_hypothesis_example}
}   
\subfigure[][]{
\includegraphics[scale=0.30]{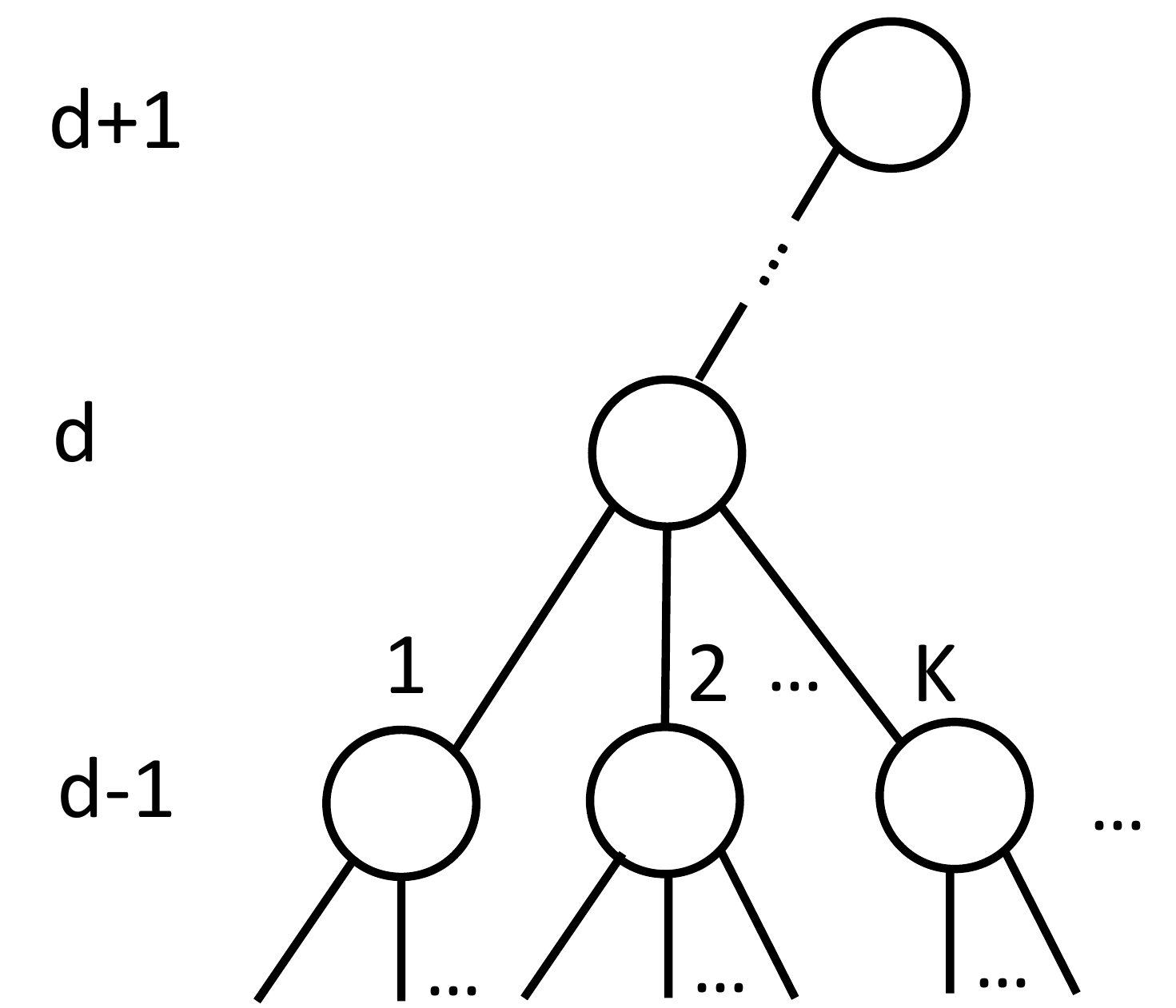}  
\label{fig:complexity-analysis-worst-case-tree}
}   
\caption{ 
\ref{fig:local_area_hypothesis_example} General network reduced to individual branches. 
 \ref{fig:local_area_hypothesis_example} Worst case tree network of depth $D$ and $K$ children for each vertex.
 } 
\end{figure}
\begin{table}[h]
\centering
\caption{ }
\begin{tabular}{@{}ccc@{}}
\toprule
Binary Flow $\mbb{I}_{ d(s_i) > 0\}}$    &  Branch Labels & Hypotheses Branches    \\
\cmidrule{1-3}  
$0, 0$ & $U$, $Z$, $U$, $U$, $Z$ & $b_1$, $b_2$, $b_2 \times b_3$,                         \\ 
           &                                          & $b_2 \times b_5$, $b_2 \times b_4 \times b_5$    \\ 
$0, 1$ & $P$, $Z$, $P$, $U$, $P$ &  $b_2 , b_2 \times b_4$                                         \\
$1, 0$ & $P$, $P$, $U$, $U$, $Z$ &  $b_3$, $b_5$, $b_4 \times b_5$                          \\
$1, 1$ & $P$, $P$, $P$, $U$, $P$ &  $b_4$                                                                    \\
\bottomrule 
\label{tab:general_local_hypothesis_solution}
\end{tabular}  
\end{table}   
An example local area is provided in Figure \ref{fig:local_area_hypothesis_example}.
The general method is applied under each of the binary flow cases, where the results are shown in Table \ref{tab:general_local_hypothesis_solution}.
The method is applied to each binary flow, and the branches to be enumerated are given. 

\subsection{Complexity Analysis}
\label{subsection-complexity-analysis}

In analyzing the complexity of various algorithms we assume the tree in Figure \ref{fig:complexity-analysis-worst-case-tree}.
Each vertex has $K$ children and is of depth $D$.
It can be shown that the number of edges is related to these quantities by $E = \frac{K^D - 1}{K-1}$.

\subsubsection{Evaluating $P^{max}_{E}(A)$}
\label{subsubsection-Evaluating-P-max-E-Area}

For simplicity, we focus on a binary tree, so $K=2$ and $|E| = 2^{D} - 1$.
The number of possible hypotheses at each depth $C(d)$ using this network is related recursively as the following:
\begin{align}
C(d+1) &= \sum_{n=1}^{K} { K \choose n } C^{n}(d+1) \\ 
            &= \left( C(d) + 1 \right)^{K} - 1
\end{align}
This can be derived directly from \eqref{eq-recursive-Hu}.

In the binary tree case, it's simple to show that $C(d) = 4^{2^d} - 1$. 
Therefore the number of hypotheses are double exponential in the depth of the tree. 
For the entire tree, this leads to the root node value of $C(D) = 4^{2^D} - 1$ which is $4^{(|E|+1)} - 1$.
Therefore $|\mc{H}_u| = O\left( 4^{|E|} \right)$, exponential in the size of the graph.

This cost may be averted due to the following:
\begin{itemize}
\item Fixed multi-hypothesis size.
The MAP detector requires a prior probability of hypotheses.  
Since multiple edge outages are less likely, they don't always have to be enumerated.
For example, considering only single edge outages leads to $|\mc{H}_u| = O\left( |E| \right)$ complexity for an area.
\item Small area sizes, and Binary Flow segmentation.
The number of hypotheses are exponential in $|E|$ for an area.
Given many sensors, this can divide the number of edges for an area considerably.
Additionally, the binary flow information from downstream sensors will on average divide each $\mc{H}^{+}(A)$ by a factor of $2^{|d(s)|}$.
\end{itemize}

The following analysis is in terms of the evaluation of $P^{max}_{E}(A)$ since we assume appropriate approximation of this function has been performed.  
\subsubsection{Evaluation of Algorithm \ref{tree-placement-algorithm} using greedy strategy} 
The greedy strategy will have to evaluate all $2^{K}$ subproblems at each vertex, and choose the minimum $P^{max}_{E}(A)$.
The worst case complexity is therefore $O(|E|~2^{K} )$, which reduces to $O(|E|)$, since $K$ is a constant.

\subsubsection{Evaluation of Algorithm \ref{tree-placement-algorithm} using optimal strategy} 
The optimal strategy will expand the problem size by a factor $2^{K}$ at each vertex.
The number of sub-problems to consider after a depth of D will be $(2^{K})^D$ which for a binary tree becomes $O(4^{\log |E|})$.
\subsection{Proxy Function Optimization}
\label{subsection-proxy-function-optimization}
The placement problem \ref{min-max-hyp-all} will output the optimal sensor locations and minimizing maximum missed detection error $\alpha^{\star}(\mc{M}^{\star})$, as:
\begin{align*}
\alpha^{\star}(\mc{M}^{\star}) = \underset{|\mc{M}|=M}{\min}~ \max_{H \in H_u } P_{E}( H, \mc{M} ) 
\end{align*}
This can be approximated by using the decoupling of hypotheses in different areas.
Recall the solution to \ref{min-max-hyp-all} repeated above, is the distributed detector in Algorithm \ref{alg-decentralized-ML-detector} where each area performs a local hypothesis $\hat{H}, \hdots, \hat{H}_M$.
The complete hypothesis is only correct if every local detection output is correct.
So for any hypothesis, we have following lower bound:
\begin{align}
\min_{H \in H_u} Pr( \hat{H} = H;& \mc{M} ) \nonumber \\
                                                   &= \min_{H \in H_u } \left( \prod_{ \forall A_i,~H_i  } Pr( \hat{H}_i = H_i; \mc{M} )  \right)         \label{pr-hu-lower-bound-line1}  \\
						 &\geq \prod_{ \forall A \in \mc{A} } \underset{H \in \mc{H}(A) }{\min} Pr( \hat{H} = H; A )         \label{pr-hu-lower-bound-line2}   \\
					 	 &= \prod_{ \forall A \in \mc{A} } P^{min}_{C}(A).   							 	    \label{pr-hu-lower-bound-line3} 
\end{align}

Line \eqref{pr-hu-lower-bound-line1} follows from the decoupling of the decentralized detector.
The overall MAP decision can be correct only if each local MAP decision is correct.
For any $H \in \mc{H}_u$ the probability of each area making a correct decision is always greater than the worst case probability of correct decision for each area.
We interchange the sensor placement and area notation since local areas are constructed from sensor placements.
Here $P^{min}_{C}(A)$ is the minimum probability of correct detection within a local area $A$.
This lower bound can be used to first upper bound the optimal $\alpha^{\star}$.
Finally, only an approximate solution to the upper bound is formulated.

\begin{align}
\alpha(\mc{M}^{\star}) &= \underset{|\mc{M}|=M}{\min}~ \max_{H \in H_u } P_{E}( H, \mc{M} )                          \label{alpha-opt-line1}   \\
	                           &= \underset{|\mc{M}|=M}{\min}~ \max_{H \in H_u } (1 - P_{C}( H, \mc{M} ) )                 \label{alpha-opt-line2}    \\
	                           &= 1 - \underset{|\mc{M}|=M}{\max}~ \min_{H \in H_u } P_{C}( H, \mc{M} )                    \label{alpha-opt-line3}    \\	                           
                                   &\leq 1 - \underset{|\mc{M}|=M}{\max}~ \prod_{ \forall A \in \mc{A} } P^{min}_{C}(A)       \label{alpha-opt-line4}    \\
                       	          &\approx 1 - \underset{|\mc{M}|=M}{\max}~\min_{ \forall A \in \mc{A} } P^{min}_{C}(A)   \label{alpha-opt-line5}    \\
                                   &= 1 -  \underset{|\mc{M}|=M}{\max}~\max_{ \forall A \in \mc{A} } (1 - P^{max}_{E}(A) )     \\                   
                                   &= \underset{|\mc{M}|=M}{\min}~\max_{ \forall A \in \mc{A} } P^{max}_{E}(A)
\end{align}
Optimization \ref{min-max-hyp-all} is identical to \eqref{alpha-opt-line2}, since the probability of a single hypothesis error can be exchanged for it's compliment.
The min-max to max-min change is due to the negative sign in \ref{alpha-opt-line3}.  
In line \ref{alpha-opt-line3}, instead of maximizing the minimum correct probability over all hypotheses we maximize a computationally tractable lower bound which is 
the product $\prod_{ \forall A \in \mc{A} } P^{min}_{C}(A)$.
In \ref{alpha-opt-line5} we introduce a close approximate solution which is the following:
Instead of maximizing the product of $P^{min}_{C}(A)$ for each A, it is sufficient to maximizing the minimum of each $P^{min}_{C}(A)$
Experimentally the two solutions have been shown identical for a large number instances, and only sub-optimal in a small number of cases where the gap is small.     
    
\begin{figure}[h]
\hspace{-4mm}
\subfigure[][]{
\includegraphics[width=0.24\textwidth, height=0.24\textwidth]{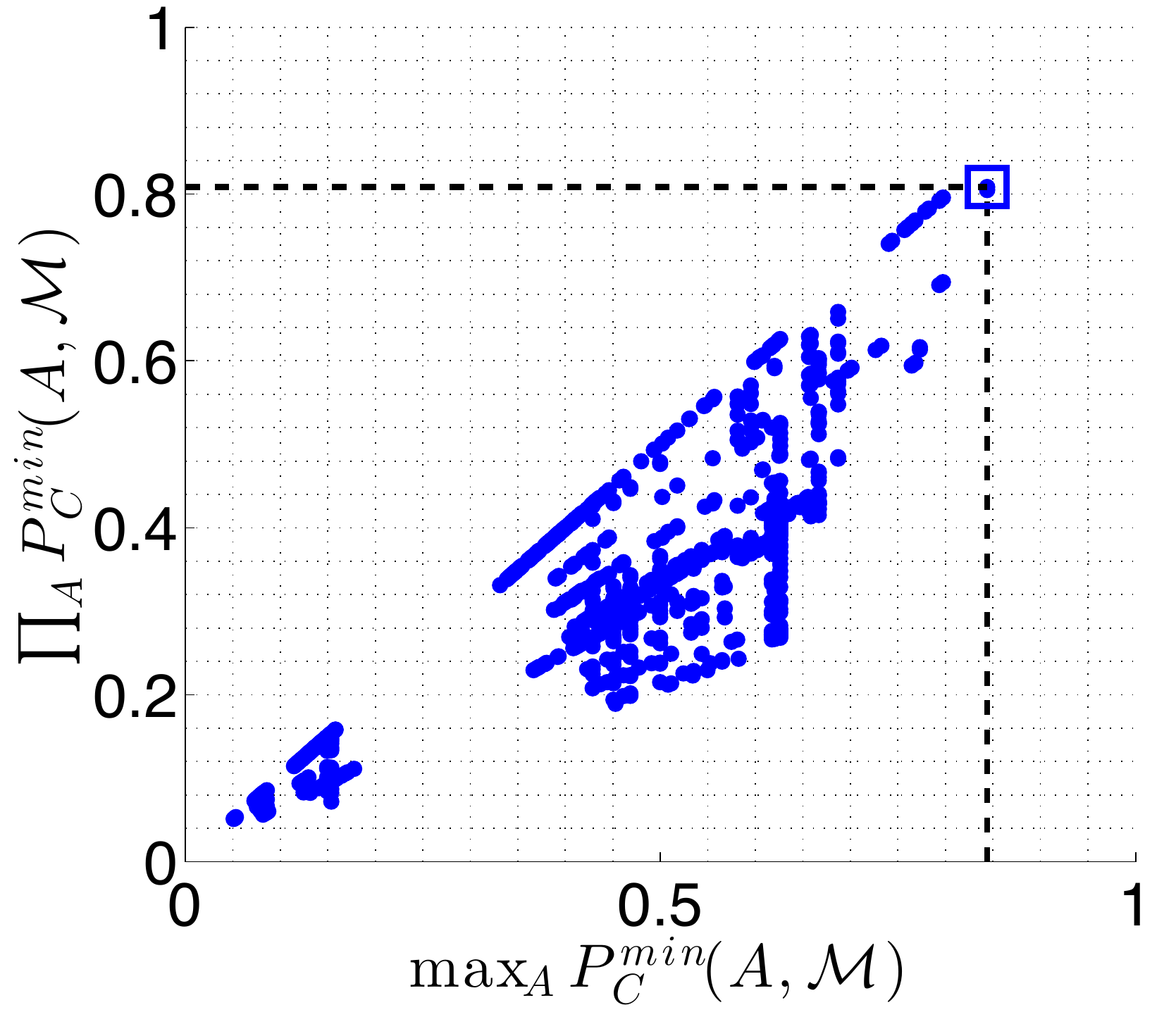} 
\label{fig:opt_approx_good_example}
}
\hspace{-4mm}
\subfigure[][]{
\includegraphics[width=0.24\textwidth, height=0.24\textwidth]{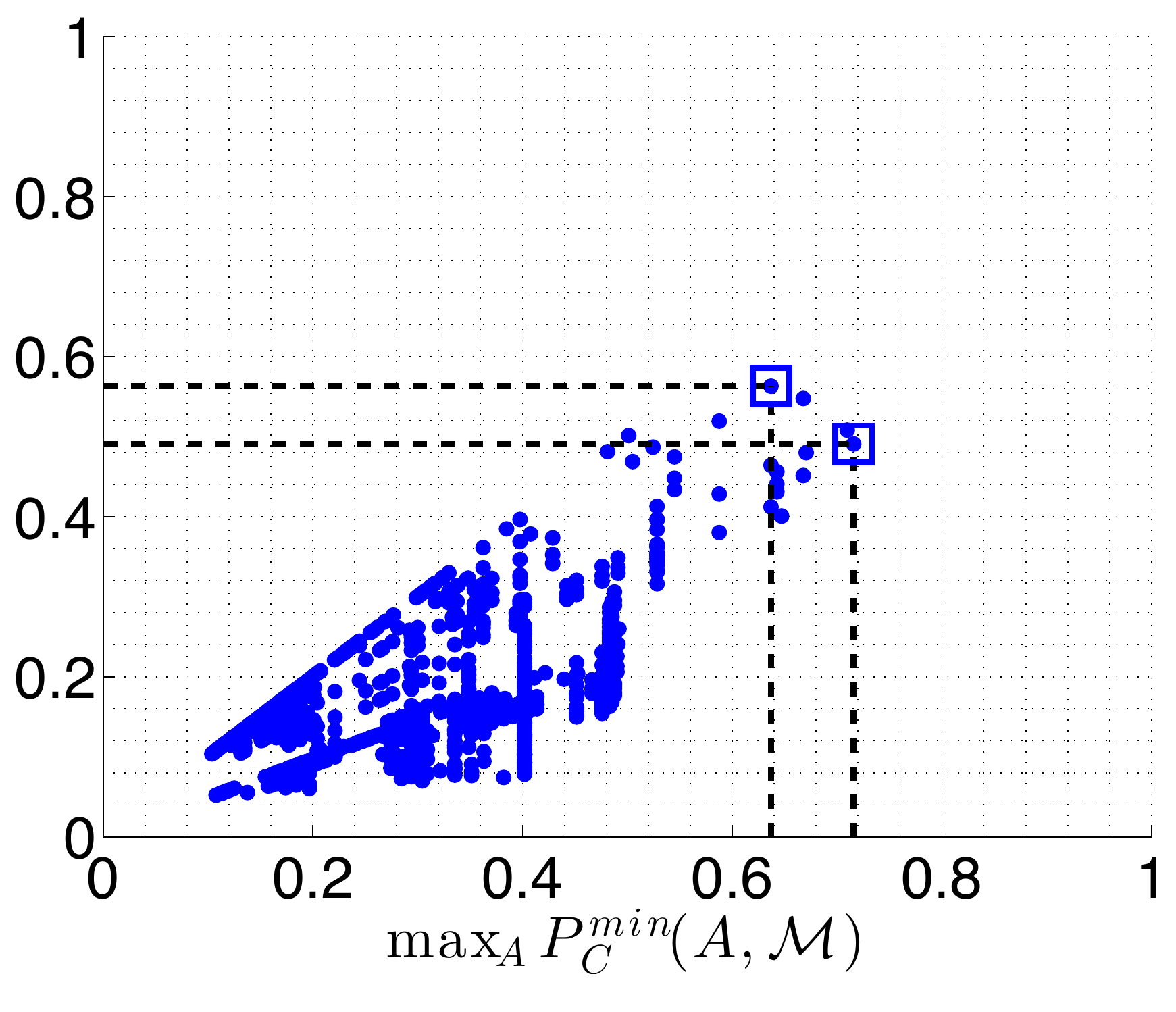} 
\label{fig:opt_approx_bad_example}
}
\caption[]{
\ref{fig:opt_approx_good_example} Brute force placement evaluation where the two objectives are identical.
\ref{fig:opt_approx_bad_example} Brute force placement evaluation where the two objectives differ. 
} 
\label{}
\end{figure}   

Experimentally the two solutions have been shown identical for a large number instances, and only sub-optimal in a small number of cases where the gap is small. 
Figure \ref{fig:opt_approx_good_example}, \ref{fig:opt_approx_bad_example}, shows a pair of randomly generated trees with $N=15$ nodes with random loads and a forecast coefficient of variation of $0.02$.
In both cases, the bottom up placement was used to determine $\mc{M}^{g}$, where $|\mc{M}^{g}| = 5$.

A brute force enumeration of all ${15 \choose 5}$ placements is evaluated for $\max_{ \forall A \in \mc{A} } P^{min}_{C}(A, \mc{M})$ and $\prod_{ \forall A \in \mc{A} } P^{min}_{C}(A, \mc{M})$.

In both cases, there is a strong correlation between the two solutions.  
{\color{black} The two solutions are not however equal
Figure \ref{fig:opt_approx_good_example}, the solutions are identical, while in \ref{fig:opt_approx_bad_example}, the two solutions differ by $7.2 \%$.
Intuitively they should intuitively be very close to each other.
Decreasing one area error will increase the other area's error due to the monotonic growth of $P^{max}_{E}(A)$ and the finite tree size.
Therefore, maximizing the product of all the terms tends to a solution where each area error is as close to each other as possible.
Minimizing the maximum error often leads to such a solution, since we must trade off one area error for another.
}
\subsection{Proof of Theorem 1}
\label{Appendix-Proof-of-Theorem-1}

We prove Theorem 1 by showing the following:
\begin{enumerate}
\item The objective function $P^{max}_{E}(A)$ monotonically increases  for nested areas.
\item Algorithm \ref{tree-placement-algorithm}, will recover the solution to \ref{feasibility-problem}
\end{enumerate}   
First we prove propositions \ref{prop-intersection-definition} and state a conjecture shown to hold in large scale simulation experiments.
These are needed to prove Lemma \ref{lemma-MMDP-non-decreasing} which is needed to prove Theorem 1.

First consider the following:
\begin{define}
For a single pairwise test, $H_{i}$ vs $H_{j}$ we have the following decision region: 
\begin{align}
r_{i,j} = \{ \mb{s} \in R^{M}: \Pr(\mb{s} | H_i ) \geq \Pr(\mb{s} | H_j ) \}.
\end{align}
\end{define}
The observation space is therefore partitioned into two regions.  
So the detector is the following:
\begin{align}
\hat{H} =    
	\begin{cases}
   	 	H_i, & \mathbf{s} \in r_{ij}        \\
    		H_j, & \mathbf{s} \notin r_{ij}.  \\
	\end{cases}
\end{align}
For the one-to-many ML test: $H_i$ vs $ \forall H_j \in H$, we have an acceptance region defined as:
\begin{align}
R_{\mc{H}}(H_i) = \{ \mb{s} \in R^{M}: \Pr(\mb{s} | H_i ) \geq \Pr(\mb{s} | H_j )~\forall H_j \in \mc{H}\}.
\end{align}

\begin{lem}
\label{prop-intersection-definition} 
An equivalent definition is 
\begin{align}
R(H_i) = \bigcap_{ j: H_j \in \mc{H} } r_{i,j}.
\end{align}
\end{lem}

\begin{proof} Using the definition of the right hand side, we have
\begin{align*}
\bigcap_{ j \in H } r_{i,j} & = \{ s: s \in r_{i,1} \cap \hdots \cap s \in r_{i,N}  \} 		                          \\
                   	            & = \{ s: \Pr(\mb{s} | H_i ) \geq \Pr(\mb{s} | H_1 ) \cap  \hdots	                          \\ 
		                    &~~~~~~~~~~ \cap \Pr(\mb{s} | H_i ) \geq \Pr(\mb{s} | H_N )   \}                       \\
        				    & = \{ s: \Pr(\mb{s} | H_i ) \geq \Pr(\mb{s} | H_j )~\forall H_j \in \mc{H} \}           \\				
			            & = R_{\mc{H}}(H_i). 										 
\end{align*}
\end{proof}

Note, that this statement is for the maximum likelihood classifier.
Under an arbitrary classifier, a procedure of constructing one-to-many classifier will lead to ambiguity.
See \cite{bishop2006} (pg. 183) for discussion.

\begin{conj}   
\label{conjecture-ML-error-monotonic-variance}
Given a ML detection problem with a set of hypothesis of the form: $s_{k} \sim N(\mu_k, \sigma^2_k + \Delta )$ for $k=1, \hdots, K$.
The missed detection error for each hypothesis will monotonically increase w.r.t $\Delta$.
\end{conj}
This is seen to hold with 200,000 random problem instantiations.
We now state Lemma \ref{lemma-MMDP-non-decreasing}.

\begin{define}
Two area networks are nested $A \subset A^{\prime}$ if the vertices of each area $V$, $V^{\prime}$ are such that $V \subset V^{\prime}$.  
\end{define}   
\begin{lem}
\label{lemma-MMDP-non-decreasing}
Given two area networks $A$ and $A^{\prime}$ where $A \subset A^{\prime}$, $P^{max}_{E}(A) \leq P^{max}_{E}(A^{\prime})$. 
\end{lem}
  
\begin{figure}[h]
\centering
\subfigure[][]{
\label{fig:AN_growth_case_1}  
\includegraphics[width=0.2\textwidth, height=0.2\textwidth]{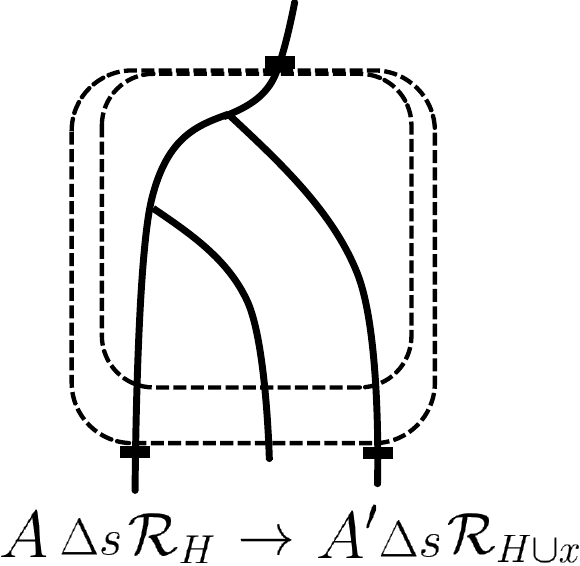} 
}
\subfigure[][]{
\label{fig:AN_growth_case_2}
\includegraphics[width=0.2\textwidth, height=0.2\textwidth]{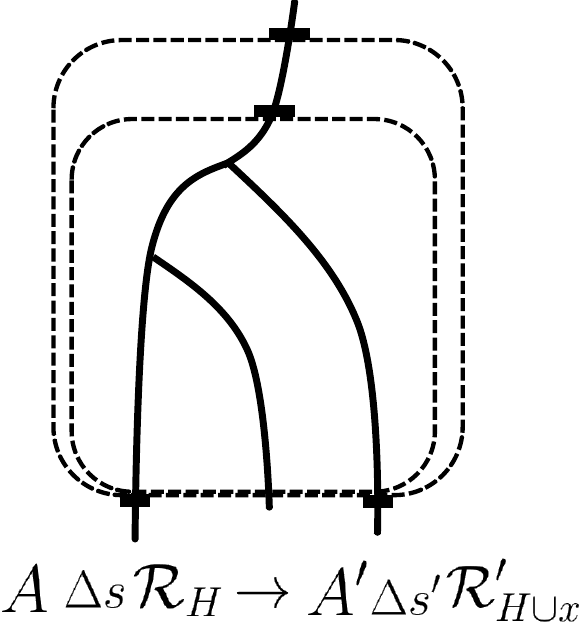} 
}    
\caption[]{
\ref{fig:AN_growth_case_1} 
Case 1 showing growth by adding new nodes by moving \textit{terminal sensor down}. 
The conditional pdf $\Delta s | H$ $\forall h \in H$ does not change.
However the acceptance region shrinks from $R_\mc{H}$ to $R_{\mc{H} \cup x}$.
\ref{fig:AN_growth_case_2} 
Case 2 showing growth by adding new nodes by moving \textit{root sensor up}. 
The conditional pdf $\Delta s^{\prime} | H$ $\forall h \in H$ will change.
Again, the acceptance region shrinks from $R_{\mc{H}}$ to $R^{\prime}_{H \cup x}$.
} 
\label{fig:AN_growth_case}
\end{figure}   
  
\begin{proof}
Given some $\alpha^{*} = P_{\text{err}}^{\max}(A)$, which is the maximum missed detection probability of the local hypothesis in area $A$ which is evaluated at some hypothesis $H^{*}$.
We aim to show that for an enlarged area $A^{\prime}$, the error probability for $H^{\star}$ will \textit{always} be larger, which is quite intuitive.
Therefore the maximum error $P^{\max}_{E}(A^{\prime})$ regardless if $H^{*}$ maximizes the missed detection in the enlarged area.

Expansion of A is analyzed in two cases:
\begin{enumerate}
\item [] \textbf{Case 1} Terminal sensors expand downstream (away from $v_0$) shown in Figure \ref{fig:AN_growth_case_1}.
\item [] \textbf{Case 2} Root sensor of $A$ moves upstream (closer to $v_0$) shown in Figure \ref{fig:AN_growth_case_2}.
\end{enumerate}  

We use the following shorthand: 
Hypotheses for area $A$ and $A^{\prime}$: $\mc{H} \triangleq \mc{H}(A)$, $\mc{H} \cup x_i \triangleq \mc{H}(A^{\prime})$.
The set $x_i$ differ in how the area is enlarged, where $i=1,2$ for case 1 and case 2.
Therefore under case 1, and 2 we have $R_{\mc{H} \cup x_1}(H_i)$ and  $R_{\mc{H} \cup x_2}(H_i)$.
Now we show how the effective measurement distribution (as defined in \eqref{eq:definition-effective-measurement}) and acceptance regions changes under each case.
\begin{enumerate}
\item [] \textbf{Case 1} For all hypothesis $H_i \in \mc{H}$ we have that $\mu^{\prime}_i = \mu_i$ and ${\sigma^{\prime}}^{2}_i = \sigma^{2}_i$.
The distribution $\Delta s | H_i$ is unchanged.
Given that $R_{\mc{H}}(H_i) = \bigcap_{j:H_j \in \mc{H}} r_{ij}$, the new acceptance region is $R_{\mc{H} \cup x_1}(H_i) = \bigcap_{j:H_j \in \mc{H} \cup x_1} r_{ij}$, with $r_{ij}$ from the original alternatives unchanged.
\item [] \textbf{Case 2} For all hypothesis $H_i \in \mc{H}$ we have that $\mu^{\prime}_i = \mu_e + \mu_i$ and ${\sigma^{\prime}}^{2}_i = \sigma^2_e + \sigma^{2}_i$.  
The distribution and acceptance region change $\Delta s | H_i \rightarrow \Delta s^{\prime} | H_i$.
The new acceptance region is the following: $R_{\mc{H} \cup x_2}(H_i) = \{ \bigcap_{j:H_j \in \mc{H}} r^{\prime}_{ij} \} \bigcap \{ \bigcap_{j:H_j \in \mc{H}} r_{ij} \}$, where the acceptance regions under the previous area alternatives are now different.
\end{enumerate}
To see why the distribution $\Delta s | H_i$ changes, first recall that:
\begin{align}
\Delta s_i |H_k &\sim N\left(\sum_{v \in V_i \setminus V_k}^{} \mu(v) - \mu_{T}, \sum_{v \in V_i \setminus V_k}^{} \sigma^2(v)  - \sigma^{2}_{T} \right) \\ 
	                &\sim N\left(\mu_k,  \sigma^{2}_k\right)  \forall H \in \mc{H}  
\end{align}

The terms $\mu_T$ and $\sigma^2_T$ are the sum of loads forecasts and variances of all terminal sensors.
Now moving the root node upstream leads to:
\begin{align}
s^{\prime}_i |H_k &\sim N\left(\sum_{v \in V^{\prime}_i \setminus V_k}^{} \mu(v) - \mu_{T}, \sum_{v \in V^{\prime}_i \setminus V_k}^{} \sigma^2(v) - \sigma^{2}_{T}  \right) \\
			  &\sim N\left( \mu_e + \mu_i,  \sigma^2_e + \sigma^{2}_i\right) \forall H \in \mc{H}.
\end{align}
Therefore changing the position of $s_i$ so as to add additional vertices will increase every original hypothesis mean and variance by the same amount.

Now consider Case 1 first, where we merely add new alternatives, keeping the distributions of $\Delta s | H$ unchanged.  
Here we have:
\begin{align}    
\Pr&(  \Delta s  \in R_{\mc{H}\cup x} (H^{*}) | H^{*} \text{ true} ) 								                                   	        \\
		 	 &= \Pr( \Delta s \in \bigcap_{j:H_j \in \mc{H} \cup x} r_{i,j}~|~H^{*}~\text{true})            \label{eq:lemma1_case1_ln1}   \\
			 &= \Pr ( \{ \Delta s \in \bigcap_{j:H_j \in \mc{H} } r_{i,j}  \}                       			            
		          \cap \{ \bigcap_{j:H_j \in x } \Delta s \in r_{i,j} \} |~H^{*}~\text{true})                    		\label{eq:lemma1_case1_ln2}   \\
		         & \leq \Pr ( \Delta s \in \bigcap_{H_j \in \mc{H}} \in r_{i,j}  |H^{*}~\text{true})     	     	\label{eq:lemma1_case1_ln3}   \\
   		         &= \Pr (  \Delta s \in \mc{R}_{\mc{H}}(H^{*}) | H^{*}~\text{true}).     	                         \label{eq:lemma1_case1_ln4}   
\end{align}    
Line \ref{eq:lemma1_case1_ln1} defines the area $R_{\mc{H}}(H^*)$ using proposition \ref{prop-intersection-definition}.
This is split into the the intersection of two separate events using our definitions of $\mc{H}$ and $\mc{H} \cup x$.
Next we use the fact that $\Pr( A \cap B) \leq \Pr(A)$.
Therefore if $\Pr ( \Delta s \in R_{\mc{H}}(H^{*}) | H^{*}~\text{true})$ $\leq$ $\Pr (  \Delta s \in R_{\mc{H} \cup x_2}(H^{*}) | H^{*}~\text{true})$ then $P^{max}_{E}(A) \leq P^{max}_{E}(A^{\prime})$.
  
We next prove Case 2 where not only are more alternatives considered for $H^{*}$, but $\Delta s | H^{*}$ is translated by fixed amount $\mu_e$, $\sigma^2_e$ in mean and variance.
This implies that:
\begin{align}
 \Pr(  \Delta s^{\prime}  \in & R_{\mc{H} \cup x} | H^{*}~\text{true} )   \nonumber \\ 
                  &   \leq \Pr(  \Delta s^{\prime} \in R_{\mc{H}}(H^{*})  | H^{*}~\text{true} )                     \label{eq:lemma1_case2_ln1}  \\
		 &= \Pr(  \Delta s^{\prime} - \mu_e   \in R_{\mc{H}}(H^{*}) - \mu_e | H^{*}~\text{true} )    \label{eq:lemma1_case2_ln2}  \\
		 &\leq \Pr(  \Delta s   \in R_{\mc{H}} | H^{*}~\text{true} ). 			                               \label{eq:lemma1_case2_ln3}					                 
\end{align}
     
The inequality in line \ref{eq:lemma1_case2_ln1} uses the identical procedure in \ref{eq:lemma1_case1_ln1} - \ref{eq:lemma1_case1_ln3}.
In line \ref{eq:lemma1_case2_ln2} we are merely shifting the gaussian $\Delta s$ and the acceptance region by $\mu_e$ using a shorthand notation.
This can be done since the MAP test is scalar.
Finally the inequality in line \ref{eq:lemma1_case2_ln3}, follows from conjecture \ref{conjecture-ML-error-monotonic-variance}.
\end{proof}

We can now prove Theorem \ref{tree-network-optimality}.

\begin{proof}
The bottom up solution $\mc{M}^{g}$ moved to the root node enlarging each area network so that each $P^{max}_{E}(A) < P^{target}$ but any further up will violate the target area.
Consider some other method produces a solution $\mc{M}^{\prime}$ which minimizes the number of sensors the error constraint, where $|\mc{M}^{\prime}| < |\mc{M}^{g}|$.
This implies that some area must increase in size, as compared to the $\mc{M}^{g}$ solution, and from Lemma \ref{lemma-MMDP-non-decreasing}, some area will violate the error constraint.
\end{proof}

\subsection{Greedy and Optimal Tree Action}
\label{appendix-greedy-optimal-tree-discussion}      

{\bf optimal}-{\bf tree}-{\bf action}:  The correct action at a node-junction is to enumerate each $2^{| {\bf child}(v_r)| - 1}$ possible trees and process them until only one remains closest to the root.
In the example in Figure \ref{fig:tree_action}, we must process to the root node processing both $A_2$ and $A_3$ in parallel as separate problem instances, with its own $v_t$ and $M^{g}$.
This is in contrast to the greedy strategy that chooses one placement and moves on.
Each problem instance is then processed, until the area objective function violates  $P^{\text{target}}$.
All but one problem instance is kept; the one where $v_t$ was closest to the root vertex.

It turns out that the {\bf greedy}-{\bf tree}-{\bf action} and the {\bf optimal}-{\bf tree}-{\bf action} procedures are in practice extremely close as discussed in Appendix \ref{appendix-greedy-optimal-tree-discussion}.
Algorithm \ref{tree-placement-algorithm} can only implement this technique, since we do not grow the search space with multiple bottom up scenarios.

\begin{figure}[h]
\subfigure[][]{
\includegraphics[scale=0.55]{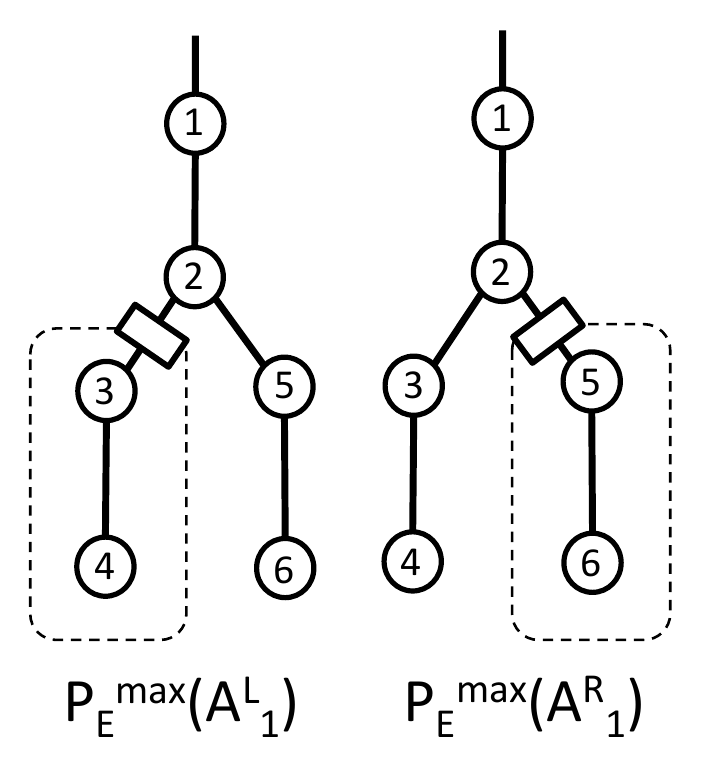}  
\label{fig:greedy_vs_optimal_counterexample_1}
}
\subfigure[][]{
\includegraphics[scale=0.55]{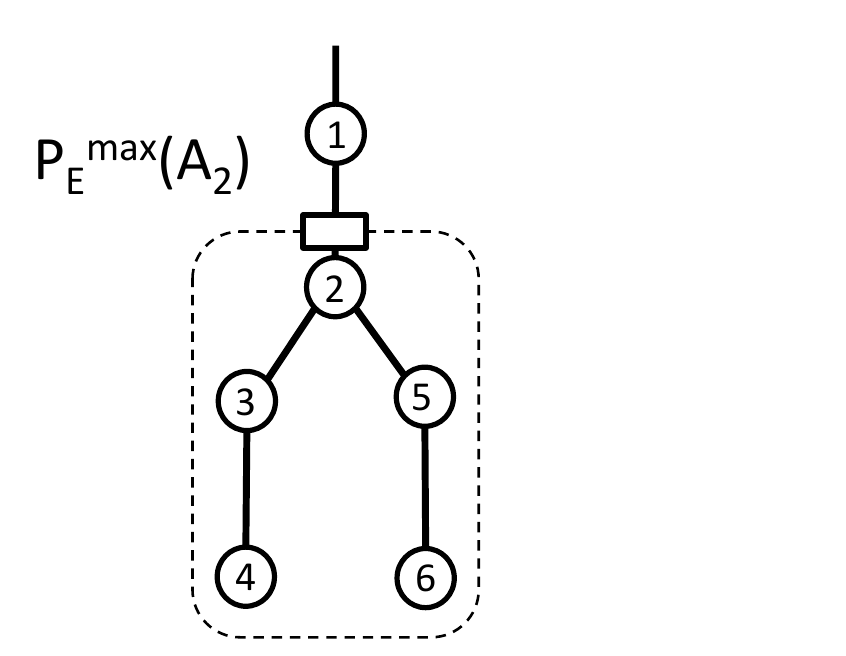}  
\label{fig:greedy_vs_optimal_counterexample_2}
}
\hspace{-5mm}
\subfigure[][]{
\includegraphics[scale=0.55]{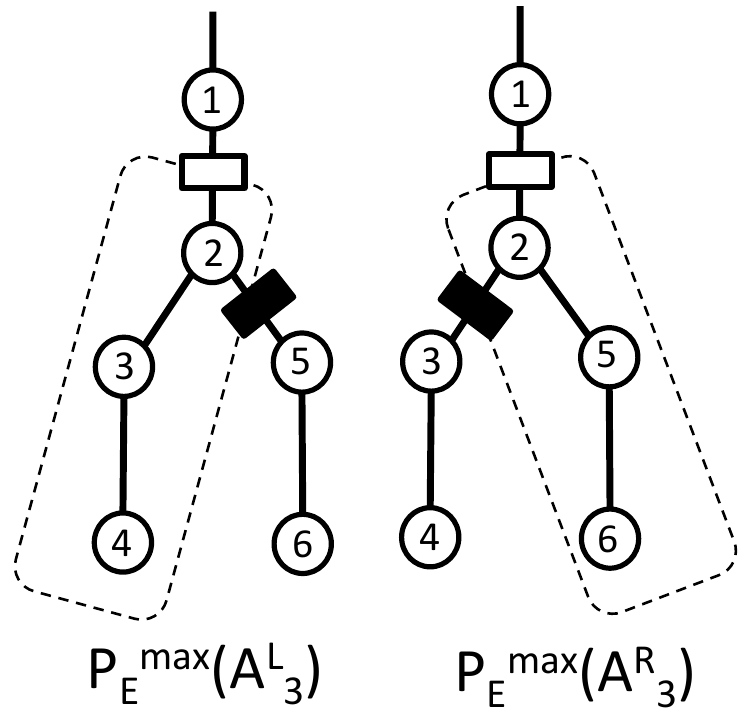}  
\label{fig:greedy_vs_optimal_counterexample_3}
}
\subfigure[][]{
\includegraphics[scale=0.55]{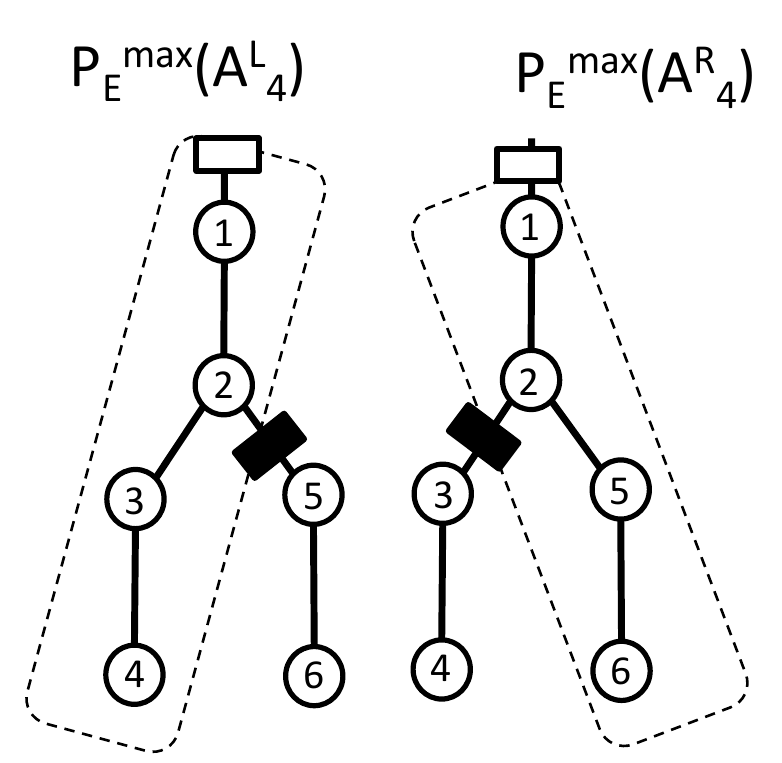}  
\label{fig:greedy_vs_optimal_counterexample_4}
}
\caption{   
\ref{fig:greedy_vs_optimal_counterexample_1} Both areas have maximum error smaller than target error: $P^{max}_{E}(A^{L}_{1}) < P^{target}$, $P^{max}_{E}(A^{R}_{1}) < P^{target}$.
\ref{fig:greedy_vs_optimal_counterexample_2} The area's are combined, and tested where $P^{max}_{E}(A_{2})  > P^{target}$.
\ref{fig:greedy_vs_optimal_counterexample_3} The greedy choice will choose $A^{L}_{3}$ as the candidate:  $P^{max}_{E}(A^{L}_{3}) < P^{max}_{E}(A^{R}_{3})  < P^{target}$. 
\ref{fig:greedy_vs_optimal_counterexample_4} The correct choice is $A^{R}_{4}$ since $P^{max}_{E}(A^{R}_{3}) < P^{target} < P^{max}_{E}(A^{L}_{3})$. }
\label{fig:greedy_vs_optimal_counterexample}
\end{figure}

The bottom up placement using {\bf greedy}-{\bf tree}-{\bf action} relies on moving the current node as close to root as possible while keeping the error $< P^{\text{target}}$.
At a juncture, recall that the current network is chosen as the one which minimizes the maximum error, and continues with that choice.
To see why this is sub optimal, consider Figure \ref{fig:greedy_vs_optimal_counterexample} shows a typical subtree where we have the following events:
\begin{enumerate}
\item  (Figure \ref{fig:greedy_vs_optimal_counterexample_1}) The bottom up algorithm will evaluate $P^{max}_{E}(A^{L}_{1}) < P^{target}$ and $P^{max}_{E}(A^{R}_{1}) < P^{target}$.
\item (Figure \ref{fig:greedy_vs_optimal_counterexample_2})
Move onto evaluating the combined area network with the parent node. 
Evaluating $P^{max}_{E}(A_2)  > P^{\text{target}}$, we must choose in a greedy manner via {\bf greedy}-{\bf tree}-{\bf action}.
\item 
(Figure \ref{fig:greedy_vs_optimal_counterexample_3})
Evaluating both $P^{max}_{E}(A^{L}_3)$ and $P^{max}_{E}(A^{R}_3)$, where $P^{max}_{E}(A^{L}_3) < P^{max}_{E}(A^{R}_3) < P^{target}$.
The greedy choice will keep $A^{L}_3$ and discard the $A^{R}_3$.
\item 
(Figure \ref{fig:greedy_vs_optimal_counterexample_4})
The optimal choice is in fact $A^{R}_4$ since $P^{max}_{E}(A^{R}_4) < P^{target} < P^{max}_{E}(A^{R}_4)$.
\end{enumerate}
The greedy choice will choose $A^{L}_3$ and be forced to place a sensor in $A^{L}_4$, while the optimal can continue upstream.
The following numerical example will lead to this: $\mu_i = 1, \forall i$ and $\sigma^2 = \{ 0.0599,~0.0125,~0.0835,~0.0945,~0.0906,~0.0607\}$ with $P^{target} = 0.1923$.

The triplet of scalar hypotheses which cause this are shown in Table \ref{tab:greedy_optimal_countexample}.
\begin{table}[h]
\centering
\caption{PNNL Test feeders used in case study}
\begin{tabular}{@{}cccccc@{}}
\toprule
Area                  &  Flow                           &  $\mc{H}^{+}(A, b)$        &  $\mu_k$        & $\sigma^2_k$                & $P^{max}_{E}(A)$     \\
\cmidrule{1-6} 
$A^{L}_{3}$       &  $s_2 > 0$,     &  $e_3, e_4, \emptyset$    &   $1,~2,~3$     & $0.0125,~ 0.1031,~0.1637$   &  0.1083   \\
                          &    $s_5 > 0$    &                                           &                        &                                                &               \\          
$A^{R}_{3}$       & $s_2 > 0$,     &   $e_5, e_6, , \emptyset$  &   $1,~2,~3$     & $0.0125 ,~0.0960,~0.1905$   &   0.0961  \\
                          &  $s_3 > 0$      &                                           &                        &                                                &                \\
$A^{L}_{4}$       &  $s_1 > 0$,     &  $e_3, e_4, \emptyset$      &   $2,~3,~4$     & $0.0724,~0.1558,~0.2504$   &  0.1885   \\
                          &   $s_3 > 0$     &                                            &                       &                                                &               \\
$A^{R}_{4}$      & $s_1 > 0$,      &   $e_5, e_6, , \emptyset$   &  $2,~3,~4$     & $0.0724,~0.1630,~0.2236$   &  0.1960   \\
                         &    $s_3 > 0$     &                                            &                      &                                                 &                \\
\bottomrule 
\label{tab:greedy_optimal_countexample}
\end{tabular}  
\end{table}  

\begin{figure}[h]
\centering
\subfigure[][]{
\includegraphics[scale=0.4]{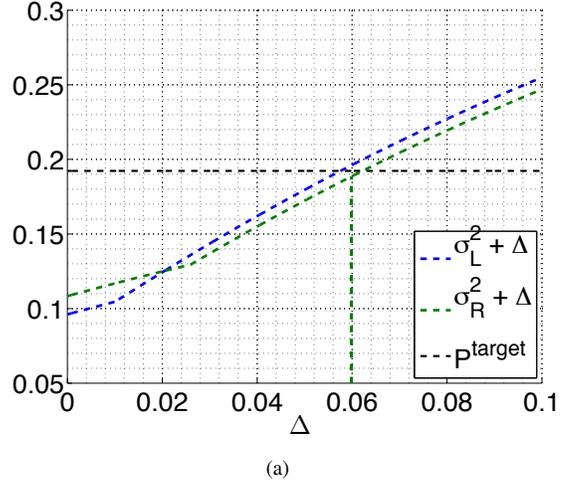}  
\label{fig:greedy_vs_optimal_counterexample_delta_error}
}  
\caption{
The maximum probability of error over three hypotheses as a function of a translation of each hypothesis variance.
At the point $\Delta = 0$, $P(A^{L}_{3} ) < P(A^{R}_{3} )$, and in point $\Delta = 0.0599$, $P(A^{L}_{4} ) > P(A^{R}_{4} )$
}
\end{figure}

Notice that in the hypothesis means and variances, we have the following:
\begin{align*}
\mu_k(A^{L}_4)      &= \mu(v_1) + \mu_k(A^{L}_3)     \\ 
\mu_k(A^{R}_4)      &= \mu(v_1) + \mu_k(A^{R}_3)     \\
\sigma^2(A^{L}_4)  &= \Delta + \sigma_k(A^{L}_3)   \\ 
\sigma^2(A^{R}_4) & = \Delta  + \sigma_k(A^{R}_3) \\ 
\end{align*}
Where $\mu(v_1) = 1$, and $\Delta  = \sigma^2(v_1) = 0.0599$. 
The translation of mean and variance causes the maximum error over an area to switch from $A^{L}$ to $A^{R}$.
Recall that conjecture \ref{conjecture-ML-error-monotonic-variance} stated that the maximum error in such a case will monotonically increase with respect to some translation in variances.
Regardless, there is no domination between any pair of triplets, whereby one will always be greater than another under the same translation.

This can be seen in Figure \ref{fig:greedy_vs_optimal_counterexample_delta_error} where the maximum hypothesis error between the tuples is shown with respect to translation $\Delta$.
As indicated, any $\Delta > 0.02$ will cause their ordering to change, with the counterexample shown to be beyond this point.

However, we should not that although the transition does occur, the gap is quiet small so any realistic gap will be very small.
Since $P^{target}$ is between  $P(A^{L}_{4} )$  and $P(A^{R}_{4} )$ it is very unlikely to occur.
In finding a counterexample $10000$ monte carlo runs produced $5$ examples.
For this reason, the greedy and optimal placement strategies will have identical outcomes of random instances.
\end{document}